\documentclass[10pt,twocolumn,twoside]{IEEEtran} 

\usepackage{graphicx}
\usepackage{subcaption}
\usepackage{caption}
\usepackage{amsmath}
\usepackage{amssymb}
\usepackage{amsthm} 
\usepackage{accents}
\usepackage{statex}
\usepackage{cases}
\usepackage{mathrsfs}
\usepackage{tikz,pgfplots}
\pgfplotsset{compat=newest}
\usepackage{xcolor}
\usepackage{caption}
\usepackage{setspace}
\usepackage{booktabs}
\usepackage{array}
\usepackage{cite}
\usepackage{enumitem}
\usepackage{algorithm}
\usepackage{algorithmic}
\usepackage{multirow}
\usetikzlibrary{patterns}
\usetikzlibrary{graphs}
\usetikzlibrary{graphs.standard}
\usetikzlibrary{matrix,arrows,fit,backgrounds,mindmap,plotmarks,decorations.pathreplacing}
\newcolumntype{C}[1]{>{\centering\arraybackslash}p{#1}}

\usetikzlibrary{automata}

\newtheorem{proposition}{Proposition}
\newtheorem{theorem}{Theorem}

\newtheorem{lemma}{Lemma}
\newtheorem{corollary}{Corollary}
\newtheorem{remark}{Remark}

\newlength\figureheight
\newlength\figurewidth

\DeclareMathOperator{\1}{\textbf{1}}

\DeclareMathOperator*{\col}{col}
\DeclareMathOperator*{\tr}{tr}
\DeclareMathOperator*{\cov}{cov}

\newcommand{\revise}[1]{\textcolor{black}{#1}}
\newcommand{\modify}[1]{\textcolor{blue}{#1}}

\tikzstyle{sensor} = [draw, fill=blue!20, rectangle, rounded corners,
minimum height=2em, minimum width=7em]
\tikzstyle{est} = [draw, fill=orange!20, rectangle, rounded corners,
minimum height=2em, minimum width=7em]
\tikzstyle{pinstyle} = [pin edge={to-,thin,black}]

\title{A Distributed Implementation of Steady-State Kalman Filter}

\author{Jiaqi Yan, Xu Yang, Yilin Mo$^*$, and Keyou You
	\thanks{
		The authors are with Department of Automation, BNRist, Tsinghua University. Emails: {jiaqiyan@tsinghua.edu.cn, xu-yang16@mails.tsinghua.edu.cn, ylmo@tsinghua.edu.cn,youky@tsinghua.edu.cn}
	}
\thanks{
			$*$: Corresponding author.
}
}

\begin{document}
\maketitle

\begin{abstract}
This paper studies the distributed state estimation in sensor network, where $m$ sensors are deployed to infer the $n$-dimensional state of a Linear Time-Invariant (LTI) Gaussian system. By a lossless decomposition of optimal steady-state Kalman filter, we show that the problem of distributed estimation can be reformulated as the synchronization of homogeneous linear systems. Based on such decomposition, a distributed estimator is proposed, where each sensor node runs a local filter using only its own measurement, alongside with a consensus algorithm to fuse the local estimate of every node. We prove that the average of estimates from all sensors coincides with the optimal Kalman estimate, and under certain condition on the graph Laplacian matrix and the system matrix, the covariance of estimation error is bounded and the asymptotic error covariance is derived. As a result, the distributed estimator is stable for each single node. We further show that the proposed algorithm has a low message complexity of $\min(m,n)$. Numerical examples are provided in the end to illustrate the efficiency of the proposed algorithm.
\end{abstract}


\section{Introduction}
The past decades have witnessed remarkable research interests in multi-sensor networked systems. As one of its important focuses, distributed estimation has been widely studied in various applications including robot formation control, environment monitoring, spacecraft navigation (see \cite{subbotin2009design,xie2012fully,luo2005universal,vu2015distributed,jia2016cooperative}). Compared with the centralized architecture, it provides better robustness, flexibility and reliability. 

One fundamental problem in distributed estimation is to estimate the state of an LTI Gaussian system using multiple sensors, where the well-known Kalman filter provides the optimal solution in a centralized manner \cite{anderson2012optimal}. Thus, many research efforts have been devoted to the distributed implementation of Kalman filter. For example, in an early work \cite{bar1986effect}, the authors suggest a fusion algorithm for two-sensor networks, where local estimate of the first sensor is considered as a pseudo measurement of the second one. Due to its ease of implementation, this approach has then inspired the sequential fusion in multi-sensor networks \cite{kim1994development,sun2004multi,chen2019distributed}, where the multiple nodes repeatedly perform the two-sensor fusion in a sequential manner. As the result of serial operation, these algorithms require special communication topology which should be sequentially connected as a ring/chain. In \cite{olfati2007distributed}, Olfati-Saber \textit{et. al} consider the more general network topology. They introduce the consensus algorithms into distributed estimation and propose Kalman-Consensus Filter (KCF), where the average consensus on local estimates is performed. \revise{Since then, various consensus-based distributed estimators have been proposed in literature \cite{olfati2005distributed,olfati2009kalman,battistelli2014consensus,li2011consensus,battistelli2016stability,del2009distributed,kar2010gossip,ma2016gossip,cattivelli2010diffusion,hu2011diffusion,cattivelli2008diffusion,farina2010distributed,haber2013moving}. For example, instead of doing consensus on local estimates, \cite{battistelli2014consensus} suggests achieving consensus respectively on noisy measurements and inverse-covariance matrices. On the other hand, Battistelli \textit{et. al} \cite{battistelli2014kullback} find that, by performing consensus on the Kullback-Leibler average of local probability density function, estimation stability is also guaranteed. They further prove that, if the single-consensus step is used, this approach is reduced to the well-known covariance intersection fusion rule \cite{chen2002estimation,he2018consistent}.} Since the consensus-based estimators usually require multiple consensus steps during each sampling period, they generate better estimation performance. 

\revise{In Fig.~\ref{fig:existing}, we present the general information flow of the existing consensus-based estimation algorithms, where $\Delta_i(k)$ is the information transmitted by sensor $i$ and to be fused by consensus algorithms, which could be the local estimates (\hspace{1pt}\cite{olfati2005distributed,olfati2007distributed}), measurements (\hspace{1pt}\cite{olfati2009kalman,li2011consensus,battistelli2014consensus,das2016consensus+}), or information matrices (\hspace{1pt}\cite{battistelli2014kullback,battistelli2018distributed}). It is noticed from the figure that the consensus/synchronization process is usually coupled with the local filter in these works, making it hard to analyze the performance of local estimates. Due to this fact, while the aforementioned algorithms are successful in distributing the fusion task over multiple nodes and providing stable local estimates, i.e. the error covariance is proved to be bounded at each sensor side, the exact calculation of error covariance can hardly be obtained. Moreover, the global optimality (namely, whether performance of the algorithm can converge to that of the centralized Kalman filter) is also difficult to be analyzed and guaranteed in some works.} 
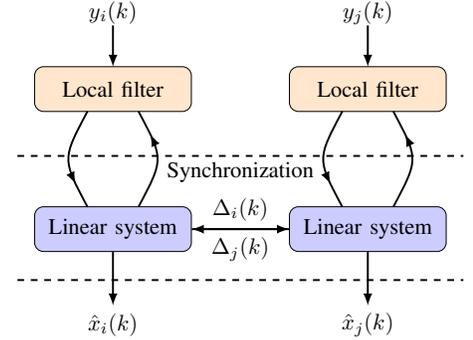
\begin{figure}
	\centering
	\resizebox{0.38\textwidth}{!}{\usetikzlibrary{decorations.markings}
\tikzset{middlearrow/.style={
        decoration={markings,
            mark= at position 0.75 with {\arrow{#1}} ,
        },
        postaction={decorate}
    }
}

\begin{tikzpicture}[auto, node distance=1.8cm]
\definecolor{lightblue}{HTML}{AFE1FD}
\definecolor{lightred}{HTML}{F2BBC3}

\node at (0,3.4) {$y_i(k)$};
\node at (4,3.4) {$y_j(k)$};

\draw [->, line width=0.3mm, -latex] (0,3.2) -- (0,2.55);
\draw [->, line width=0.3mm, -latex] (4,3.2) -- (4,2.55);

\node [sensor,align=center] (est1) {Linear system};
\node [sensor, right of=est1,node distance=4cm,align=center] (est2) {Linear system};
\node [est,above of=est1,align=center,node distance=2.2cm] (sensor1) {Local filter};
\node [est,above of=est2,align=center,node distance=2.2cm] (sensor2) {Local filter};

\draw[middlearrow={latex}, line width=0.3mm] (-0.4,1.85) .. controls(-0.8, 1.1) .. (-0.4,0.35);
\draw[middlearrow={latex}, line width=0.3mm] (0.4,0.35) .. controls(0.8, 1.1) .. (0.4,1.85);
\draw[middlearrow={latex}, line width=0.3mm] (3.6,1.85) .. controls(3.2, 1.1) .. (3.6,0.35);
\draw[middlearrow={latex}, line width=0.3mm] (4.4,0.35) .. controls(4.8, 1.1) .. (4.4,1.85);

\draw [->, line width=0.3mm, -latex] (est1) -- node {$\Delta_i(k)$} (est2);
\draw [->, line width=0.3mm, -latex] (est2) -- node {$\Delta_j(k)$} (est1);


\draw[dashed, line width=0.3mm] (-1.5, 1.15) -- (5.5, 1.15);
\draw[dashed, line width=0.3mm] (-1.5, -0.8) -- (5.5, -0.8);

\node at (2,0.9) {{Synchronization}};

\draw [->, line width=0.3mm, -latex] (est1) --  (0,-1.2);
\draw [->, line width=0.3mm, -latex] (est2) -- (4,-1.2);

\node at (0,-1.5) {$\hat{x}_i(k)$};
\node at (4,-1.5) {$\hat{x}_j(k)$};
\end{tikzpicture}}
	\caption{\revise{The information flow of most existing algorithms, where sensors $i$ and $j$ are immediate neighbors.}}
	\label{fig:existing}
\end{figure}


It is worth noticing that in theory, the gain of the Kalman filter converges to a steady-state gain exponentially fast\cite{shi2011sensor}, which can be calculated off-line. Moreover, in practice, a fixed gain estimator is usually implemented, which has the same asymptotic performance as the time-varying Kalman filter. Hence, this paper focuses on the distributed implementation of the centralized steady-state Kalman filter. 
\revise{In contrast to most of the existing algorithms, we decouple the local filter from the consensus process.} Such decoupling enables us to provide a new framework for designing distributed estimators, by reformulating the problem of distributed state estimation into that of linear system synchronization. We, hence, are able to leverage the methodologies from latter field to propose solutions for distributed estimation. To be specific, in the synchronization of linear systems, the dynamics of each agent is governed by an LTI system, the control input to which is generated using the local information within the neighborhood, in order to achieve asymptotic consensus on the local states of agents. Over the past years, lots of research efforts have been devoted to this area (see \cite{you2011network,you2013consensus,xu2019distributed,gu2011consensusability,amato2001finite,su2012two} for examples) by designing synchronization algorithms that can handle various network constraints. Exploiting the results therein, the distributed estimator in this work is designed through two phases as below:

1) (Local measurement processing) A lossless decomposition of steady-state Kalman filter is proposed, where each sensor node runs a local estimator based on this decomposition using solely its own measurement.

2) (Information fusion via consensus) The sensor infers the local estimates of all the others via a modified consensus algorithm designed for achieving linear system synchronization. 

The contributions of this paper are summarized as follows:

1) By removing assumptions regarding the eigenvalues of system matrix, this paper extends, in a non trivial way, the results in \cite{mo2016secure}, and thus develops the local filters for losslessly decomposing Kalman filter in estimating the general systems. \modify{(Lemma~\ref{lmm:decompose})}

2) Through the decomposition of Kalman filter, this paper bridges two different fields and \modify{makes it possible to leverage a general class of algorithms designed for achieving the synchronization of linear systems to solve the problem of distributed state estimation. By doing so,  we can propose stable distributed estimators under different communication constraints, such as time delay, switching topology, random link failures, \textit{etc}. (Theorem~\ref{thm:general})}

3) For certain synchronization algorithm, e.g., \cite{you2011network}, the stability criterion of the proposed estimator is established. \modify{Moreover, in contrast to the existing literature, the covariance of the estimation error can be exactly derived by solving Lyapunov equations. (Theorem~\ref{thm:optimal}, Theorem~\ref{thm:observable}, and Corollary~\ref{col:error})}

4) The designed estimator enjoys low communication cost, where the size of message sent by each sensor is $\min\{m,n\},$ with $n$ and $m$ being dimensions of the state and measurement respectively. \modify{(Remark~\ref{rmk:n<m})}

Some preliminary results are reported in our previous work \cite{2101.10689}, where most of the proofs are missing. This paper further extends the results in \cite{2101.10689} by computing the exact asymptotic error covariance, instead of only showing the stability of proposed algorithms. The extension to the more general random communication topology is also added. Moreover, a model reduction method is further proposed in this work to reduce the message complexity from $m$ to $\min\{m,n\}$. 

\textit{Notations}: For vectors $v_{i} \in \mathbb{R}^{m_{i}},$ the vector $\left[v_{1}^{T}, \ldots, v_{N}^{T}\right]^{T}$ is defined by $\col(v_{1}, \ldots, v_{N}).$ Moreover, $A\otimes B$ indicates the Kronecker product of matrices $A$ and $B$. \revise{Throughout this paper, we define a stochastic signal as ``stable" if its covariance is bounded at any time.}

The remainder of this paper is organized as follows. Section \ref{sec:formulation} introduces the preliminaries and formulates the problem of interest. A lossless decomposition of optimal Kalman filter is given in Section \ref{sec:decompose}, where a model reduction approach is further proposed to reduce the system order. With the aim of realizing the optimal Kalman filter, distributed solutions for state estimation are given and analyzed in Section \ref{sec:algorithm}. We then discuss some extensions in Section \ref{sec:discuss} and validate performance of the developed estimator through numerical examples in Section \ref{sec:simulation}. Finally, Section \ref{sec:conclusion} concludes the paper.

\section{Problem Formulation}\label{sec:formulation}
In this paper, we consider the LTI system as given below:
\begin{equation}\label{eqn:plant}
  x(k+1) = Ax(k) + w(k),
\end{equation}
where $x(k)\in\mathbb{R}^n$ is the system state, $w(k)\sim \mathcal{N} (0, Q )$ is independent and identically distributed (i.i.d) Gaussian noise with zero mean and covariance matrix $Q\geq 0$. The initial state $x(0)$ is also assumed to be Gaussian with zero mean and covariance matrix $\Sigma\geq 0$, and is independent from the process noise $\{w(k)\}$.

A network consisting of $m$ sensors is monitoring the above system. The measurement from each sensor $i\in\{1,\cdots,m\}$ is given by \footnote{The results in this paper can be readily generalized to cases where the sensor outputs a vector measurement, by treating each entry independently as a scalar measurement.}:
\begin{equation}\label{eqn:sensoroutput}
  y_i(k) = C_ix(k) + v_i(k),
\end{equation}
where $y_i(k)\in\mathbb{R}$ is the output of sensor $i$, $C_i$ is an $n$-dimensional row vector, and $v_i(k)\in\mathbb{R}$ is the Gaussian measurement noise.

By stacking the measurement equations, one gets
\begin{equation}\label{eqn:sensormatrix}
  y(k) = Cx(k) + v(k),
\end{equation}
where
\begin{equation}
  \begin{split}
    y(k) \triangleq {\left[\begin{array}{c}
	  y_{1}(k) \\
	  \vdots \\
	  y_{m}(k)
	  \end{array}\right],} \; C \triangleq {\left[\begin{array}{c}
	  C_{1} \\
	  \vdots \\
	  C_{m}
    \end{array}\right],} \;
    v(k) \triangleq {\left[\begin{array}{c}
	  v_{1}(k) \\
	  \vdots \\
	  v_{m}(k)
    \end{array}\right]},
  \end{split}
\end{equation}
and $v(k)$ is zero-mean i.i.d. Gaussian noise with covariance $R\geq0$ and is independent from $w(k)$ and $x(0)$. 

Throughout this paper, we assume that $(A,C)$ is observable. On the other hand, $(A, C_i)$ may not necessarily be observable, i.e., a single sensor may not be able to observe the whole state space. 

\subsection{Preliminaries: the centralized Kalman filter}
If all measurements are collected at a single fusion center, the centralized Kalman filter is optimal for state estimation purpose, and provides a fundamental limit for all other estimation schemes. For this reason, this part will briefly review the centralized solution given by the Kalman filter.

Let us denote by $P(k)$ the error covariance of estimate given by Kalman filter at time $k$.
Since $(A,C)$ is observable, it is well-known that the error covariance will converge to the steady state \cite{anderson2012optimal}:
\begin{align}
	P=\lim _{k \rightarrow \infty} P(k). \label{eqn:KFcov}
\end{align}

Since the operation of a typical sensor network lasts for an extended period of time, we assume that the Kalman filter is in the steady state, or equivalently $\Sigma=P$, which results in a steady-state Kalman filter with fixed gain\footnote{Notice that even if $\Sigma \neq P$, the Kalman estimate converges to the steady-state Kalman filter, i.e., the steady-state estimator is asymptotically optimal.}
\begin{align}\label{eqn:KFgain}
	K=P C^{T}\left(C P C^{T}+R\right)^{-1}.
\end{align}
Accordingly, the optimal Kalman estimate is computed recursively as 
\begin{equation}\label{eqnn:optimalest}
	\begin{split}
		\hat{x}(k+1) &= A\hat{x}(k)+K(y(k+1)-CA\hat{x}(k))\\&=(A-KCA)\hat{x}(k)+Ky(k+1).
	\end{split}
\end{equation}

It is clear that the optimal estimate \eqref{eqnn:optimalest} requires the information from all sensors. However, in a distributed framework, each sensor is only capable of communicating with immediate neighbors, rendering the centralized solution impractical. Therefore, this paper is devoted to the implementation of Kalman filter in a distributed fashion.

\section{Decomposition of Kalman Filter}\label{sec:decompose}
In this section, we shall provide a local decomposition of the Kalman filter \eqref{eqnn:optimalest}, where the Kalman estimate can be recovered as a linear combination of the estimates from local filters. \revise{This section extends, in a non-trivial way, the results in \cite{mo2016secure} by removing the assumptions on the eigenvalues of system matrix therein, and thus proposes the local filter for estimating the general systems.} The results in this part would further help us to design distributed estimation algorithms in the next sections. 

Without loss of generality, let the system matrix be
	\begin{equation}\label{eqn:diagA}
		A = \begin{bmatrix}
			A^u & \\
			& A^s
		\end{bmatrix},
	\end{equation}
	where $A^u\in\mathbb{R}^{n^u\times n^u}$ and $A^s\in\mathbb{R}^{n^s\times n^s}$, such that any eigenvalue of $A^u$ lies on or outside the unit circle while the eigenvalues of $A^s$ are strictly within the unit circle. It thus follows from \eqref{eqn:plant} that 
	\begin{equation}\label{eqn:xs}
		x^s(k+1) = A^s x^s(k)+Jw(k),
	\end{equation}
	where $J = \begin{bmatrix}
		0 & \1_{n^s}
	\end{bmatrix}\in\mathbb{R}^{n_s\times n}$ and $x(k) = \col(x^u(k), x^s(k)).$ Accordingly, $C_i$ is partitioned as
	\begin{equation}\label{eqn:diagC}
		\begin{aligned}
			C_i = \begin{bmatrix}
				C_i^u & C_i^s
			\end{bmatrix},
		\end{aligned}
\end{equation} where $C_i^u\in\mathbb{R}^{n^u\times 1}$ and $C_i^s\in\mathbb{R}^{n^s\times 1}$. 

\subsection{Local decomposition of Kalman filter}

To locally decompose Kalman filter, we first introduce the following lemmas, the proofs of which are given in appendix:
\begin{proposition}\label{prop:controllable}
	If $\Lambda$ is a non-derogatory\footnote{A matrix is defined to be non-derogatory if every eigenvalue of it has geometric multiplicity $1$.} Jordan matrix, then both $(\Lambda,\,\mathbf 1)$ and $(\Lambda^T,\mathbf 1)$ are controllable.
\end{proposition}

\begin{lemma}\label{lmm:observable}
Let $(X,p)$ be controllable, where $X\in \mathbb R^{n\times n}$ and $p \in \mathbb R^n$. For any $q\in\mathbb R^n$, if $X+pq^T$ and $X$ do not share any eigenvalues, then $(X+pq^T,q^T)$ is observable, or equivalently $(X^T+qp^T,q)$ is controllable.
\end{lemma}

\begin{lemma}\label{lmm:regular}
Let $(X,p)$ be controllable, where $X\in \mathbb R^{n\times n}$ and $p \in \mathbb R^n$. Denote the characteristic polynomial $X$ as $\varphi(s) = \det(sI-X)$. 
	Let $Y\in \mathbb R^{m\times m}$ and $q \in \mathbb R^m$. Suppose that
	\begin{equation}
		\varphi(Y)q = 0,
	\end{equation}
	then there exists $T\in\mathbb{R}^{m\times n}$ which solves the equation below:
	\begin{equation}\label{eqn:T1}
		TX=YT,\;Tp = q.
	\end{equation}
\end{lemma}

With the above preparations, let us consider the optimal Kalman estimate in \eqref{eqnn:optimalest}. For simplicity, we denote by $K_j$ the $j$-th column of the Kalman gain $K$. Namely, $K = [K_1,\cdots,K_m]$. Accordingly, \eqref{eqnn:optimalest} can be rewritten as 
\begin{equation}\label{eqn:optimalKF}
	\hat{x}(k+1) = (A-KCA)\hat{x}(k)+\sum_{i=1}^m K_iy_i(k+1).
\end{equation}

Notice that $A-KCA$ is stable. It is clear that we can always find a Jordan matrix $\Lambda\in \mathbb R^{n\times n}$, such that $\Lambda$ is strictly stable, non-derogatory and has the same characteristic polynomial of $A-KCA$. In view of Proposition~\ref{prop:controllable}, we conclude that $(\Lambda, \mathbf 1)$ is controllable. Therefore, by Lemma~\ref{lmm:regular}, we can always find matrices $F_i$'s, such that the following equalities hold for $i=1,\cdots,m$:
\begin{equation}\label{eqn:F_i}
	F_i \Lambda = (A-KCA)F_i,\;F_i \mathbf 1_n = K_i.
\end{equation}

Suppose each sensor $i$ performs the following local filter solely based on its own measurements:
\begin{equation}\label{eqn:xi}
  \hat \xi_i(k+1)=\Lambda\hat\xi_i(k)+ \1_ny_i(k+1),
\end{equation}
where $\hat\xi_i(k)$ is the output of local filter from sensor $i$, and $\1_n\in\mathbb{R}^n$ is a vector of all ones. Then it is proved the optimal Kalman filter can be decomposed as a weighted sum of local estimates $\hat \xi_i(k)$'s, as stated below:
\begin{lemma}
  \label{lmm:decompose}
  Suppose each sensor performs the local filter \eqref{eqn:xi}. The optimal Kalman estimate \eqref{eqnn:optimalest} can be recovered from the local estimates $\hat \xi_i(k), i=1,2,\cdots,m$ as
  \begin{equation}\label{eqn:localdecompose}
    \hat{x}(k) = \sum_{i=1}^{m}F_i \hat\xi_i(k),
  \end{equation}
  where $F_i$ is defined in \eqref{eqn:F_i}. 
\end{lemma}
\begin{proof}
By multiplying both sides of the recursive equation \eqref{eqn:xi} by $F_i$, we arrive at
\begin{equation}
	F_i	\hat \xi_i(k+1)= F_i \Lambda \hat \xi_i(k)+ F_i \1_n y_i(k+1).
\end{equation}
Then it follows from \eqref{eqn:F_i} that
\begin{equation}
	F_i	\hat \xi_i(k+1)= (A-KCA)F_i \hat \xi_i(k)+ K_i y_i(k+1),
\end{equation}
Summing up the above equation for all $i=1,\cdots,m$ and comparing it with \eqref{eqn:optimalKF}, we can conclude that \eqref{eqn:localdecompose} holds.
\end{proof}

\revise{Notice that the equality in Lemma~\ref{lmm:decompose} surely holds. That means the Kalman filter can be perfectly recovered by \eqref{eqn:localdecompose}. We hence claim that \eqref{eqn:localdecompose} is a lossless decomposition of optimal Kalman filter.} To better illustrate the ideas, the information flow of centralized Kalman filter and local decomposition \eqref{eqn:localdecompose} is given in Fig~\ref{fig:infoflow}.

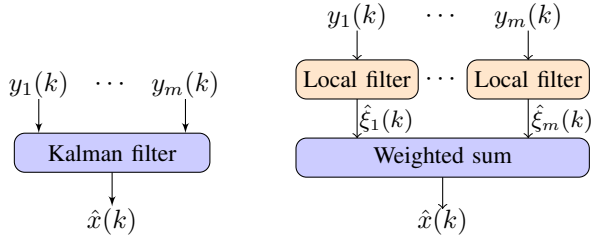
\begin{figure}[htbp]
	\begin{minipage}[b]{0.23\textwidth}
		\centering
		\begin{tikzpicture}[auto, node distance=1.8cm,>=latex',scale=0.65]

\node at (0,3.5) {$y_1(k)$};
\node at (1.5,3.5) {$\cdots$};
\node at (3,3.5) {$y_m(k)$};

\draw [->] (0,3.2) -- (0,2.5);
\draw [->] (3,3.2) -- (3,2.5);

\draw [rounded corners, fill=blue!20](-0.5,2.5) rectangle (3.5,1.7);
\node at (1.5,2.1) {{\small Kalman filter}};

\draw [->] (1.5,1.7) -- (1.5,1);
\node at (1.5,0.7) {$\hat{x}(k)$};
\end{tikzpicture}
	\end{minipage}%
	\begin{minipage}[b]{0.26\textwidth}
		\centering
		\begin{tikzpicture}[scale=0.65]

\node at (0.5,3.4) {$y_1(k)$};
\node at (2.25,3.4) {$\cdots$};
\node at (4,3.4) {$y_m(k)$};

\draw [->] (0.5,3.1) -- (0.5,2.5);
\draw [->] (4,3.1) -- (4,2.5);

\draw [rounded corners, fill=orange!20](-0.75,2.5) rectangle (1.75,1.7);
\node at (0.5,2.1) {{\small Local filter}};

\draw [rounded corners, fill=orange!20](2.75,2.5) rectangle (5.25,1.7);
\node at (4,2.1) {{\small Local filter}};

\node at (2.25,2.1) {$\cdots$};

\draw [->] (0.5,1.7) -- (0.5,0.9);
\draw [->] (4,1.7) -- (4,0.9);

\node at (1.1,1.3) {{\small $\hat{\xi}_1(k)$}};
\node at (4.7,1.3) {{\small $\hat{\xi}_m(k)$}};

\draw [rounded corners, fill=blue!20](-0.75,0.9) rectangle (5.25,0.1);
\node at (2.25,0.5) {{\small Weighted sum}};

\draw [->] (2.25,0.1) -- (2.25,-0.55);
\node at (2.25,-0.8){$\hat{x}(k)$};

%
%
%
%
\end{tikzpicture}
	\end{minipage}
	\caption{The information flow of centralized Kalman filter (left hand), and local decomposition of Kalman filter \eqref{eqn:localdecompose} (right hand).}
	\label{fig:infoflow}
\end{figure}

\subsection{A reformulation of \eqref{eqn:xi} with stable inputs}
It is noted that
 the system matrix $A$ may be unstable which implies that the covariance of measurement $y(k)$ is not necessarily bounded. As a result, we need to redesign \eqref{eqn:xi} using the stable residual $z_i(k)$ as an input instead of the raw measurement $y_i(k)$. The main reason for this reformulation is to make the consensus algorithm feasible \revise{and develop stable distributed estimators, which will be further discussed in the proof of Theorem~\ref{thm:observable}}. 
 
Towards the end, notice that $(\Lambda,1)$ is controllable, $\Lambda$ is stable and any eigenvalue of $A_u$ is unstable. Hence, we can always find a non-zero $\beta\in \mathbb{R}^n$ and compute
\begin{equation}\label{eqn:Sdef}
	S  = \Lambda + 1\beta^T,
\end{equation}
such that
\begin{enumerate}
	\item the characteristic polynomial of $A^u$ divides $\phi(s)$, where $\phi(s)$ is the characteristic polynomial of $S$, and $\phi(s) / \det(sI-A^u)$ has only strictly stable roots;
	\item $S$ do not share eigenvalues with $\Lambda$. Hence, by the virtue of Lemma~\ref{lmm:observable}, $(S^T,\beta)$ is controllable.
\end{enumerate}
\begin{remark}\label{rmk:eigenvalue}
Notice that by using $\beta$, we place the eigenvalues of $S$ to the locations which consist of two parts: the unstable ones that coincide with the eigenvalues of $A_u$ and the stable ones that are freely assigned but cannot be the eigenvalues of $\Lambda$. This is feasible as $(\Lambda,1)$ is controllable.
\end{remark} 

Next, let us consider the filter below:
\begin{equation}\label{eqn:xi_z}
  \begin{split}
    z_i(k) = y_i(k+1)-\beta^T\hat{\xi}_i(k),\\
    \hat{\xi}_i(k+1) = S \hat{\xi}_i(k)+\1_n z_i(k),
  \end{split}
\end{equation}
where $\beta$ and $S$ are calculated through \eqref{eqn:Sdef}. In the following lemma, we shall show that \eqref{eqn:xi_z} also losslessly decomposes the Kalman filter. Moreover, the covariance of $z_i(k)$ is bounded at any time.

\begin{lemma}\label{lmm:localeqv}
Consider the local filter \eqref{eqn:xi_z}. The following statements hold at any instant $k$:
\begin{enumerate}
\item \revise{\eqref{eqn:xi_z} has the same input-output relationship with \eqref{eqn:xi}. Namely, given the input $y_i(k)$, they yield the same output $\hat\xi_i(k)$;}
\item $z_i(k)$ is stable, \revise{i.e., the covariance of $z_i(k)$ is always bounded}.
\end{enumerate}
\end{lemma}
\begin{proof}
The proof is given in Appendix-\ref{sec:app_localeqv}.
\end{proof}

\begin{remark}
If $A$ has unstable modes, the previous discussions show that \eqref{eqn:xi} can be seen as a linear system with stable system matrix $\Lambda$ but unstable input $y_i(k+1)$. As a contrast, \eqref{eqn:xi_z} has unstable system matrix $S$ but stable input $z_i(k)$. \revise{This formulation is essential to guarantee the stability of local estimators, as will be seen in the proof of Theorem~\ref{thm:general}.}
\end{remark}

\subsection{\revise{A reduced-order decomposition of Kalman filter when $n<m$}}\label{sec:reduction}

To simplify notations, we define the following aggregated matrices:
\begin{equation}\label{eqn:SL}
  \begin{split}
    &\tilde{S} \triangleq I_{m} \otimes S, \;\tilde{L}_{i} \triangleq e_{i} \otimes \1_n,\; \tilde{L}\triangleq[\tilde{L}_1,\cdots,\tilde{L}_m]=I_m\otimes \1_n,
  \end{split}
\end{equation}
where $I_{m}$ is an $m$-dimensional identity matrix and $e_{i}$ is the $i$th canonical basis vector in $\mathbb{R}^m$. We thus collect \eqref{eqn:localdecompose} and \eqref{eqn:xi_z} in matrix form as:
\begin{equation}\label{eqn:fullmodel}
  \begin{split}
    \begin{bmatrix}
      \hat \xi_1(k+1) \\
      \vdots\\
      \hat \xi_m(k+1)
    \end{bmatrix}&=
    \tilde{S}
    \begin{bmatrix}
      \hat \xi_1(k) \\
      \vdots\\
      \hat \xi_m(k)
    \end{bmatrix}+
    \tilde{L}
    \begin{bmatrix}
      z_1(k) \\
      \vdots\\
      z_m(k)
    \end{bmatrix},\\
    \hat{x}(k) &= F
    \begin{bmatrix}
      \hat \xi_1(k) \\
      \vdots\\
      \hat \xi_m(k)
    \end{bmatrix}.
  \end{split}
\end{equation}
where $F \triangleq\left[F_{1}, F_{2}, \cdots, F_{m}\right]$. By Lemmas~\ref{lmm:decompose} and \ref{lmm:localeqv}, \eqref{eqn:fullmodel} represents a lossless decomposition of Kalman filter.

Notice that the system order of \eqref{eqn:fullmodel} is $mn$. In this part, we shall show that by performing model reduction, this order can be further reduced to $n^2$ when the state dimension is less than the number of sensors, namely $n<m$. These discussions would be useful for us to achieve a low communication complexity in distributed frameworks. 

\revise{To proceed, we regard the input and output of \eqref{eqn:fullmodel} as $z(k)$ and $\hat{x}(k)$, respectively, where \begin{equation}\label{eqn:z}
	z(k)\triangleq \begin{bmatrix}
		z_1(k), \cdots,z_m(k)
	\end{bmatrix}^T.
\end{equation}
}Let us introduce the below lemma, the proof of which is given in Appendix-\ref{sec:app_matrixdecomp}:
\begin{lemma}\label{lmm:matrixdecompse}
	Any matrix $W \in \mathbb{R}^{n\times n}$ can be decomposed as 
	\begin{equation}\label{eqn:Wdecompse}
		W = H_1 \varphi_1(S) + H_2 \varphi_2(S) + \cdots + H_n \varphi_n(S),
	\end{equation}
	where $H_i \triangleq e_i \beta^T$, $\{\varphi_j(S)\}$ are certain polynomials of $S$, and $S$ and $\beta$ are given in \eqref{eqn:Sdef}.
\end{lemma}

As a direct result of Lemma \ref{lmm:matrixdecompse}, for any $F_i$ in \eqref{eqn:localdecompose}, we can always rewrite it by using the polynomials of $S$, i.e., $\{p_{ij}(S)\}$:
\begin{equation}\label{eqn:Fdecompose}
	F_i= \sum_{j=1}^n H_j p_{ij}(S).
\end{equation}
For simplicity, we also denote
\begin{equation}\label{eqn:T_i}
	T_i \triangleq [(p_{i1}(S) \1_n)^T, \cdots, (p_{in}(S) \1_n)^T]^T.
\end{equation}

It is then proved in the below theorem that system \eqref{eqn:fullmodel} can be reduced with a less order:
\begin{theorem}\label{lmm:mdl_reduce}
  Consider the following system: 
  \begin{equation}\label{eqn:aft_mdl_reduce}
    \begin{split}
      \begin{bmatrix}
	\theta_1(k+1)\\\vdots\\\theta_n(k+1)
      \end{bmatrix}&=(I_n\otimes S)
      \begin{bmatrix}
	\theta_1(k)\\\vdots\\\theta_n(k)
      \end{bmatrix}+ T
      \begin{bmatrix}
	z_1(k)\\\vdots\\z_m(k)
      \end{bmatrix},\\
      \tilde{x}(k) &=\revise{H 
      \begin{bmatrix}
	\theta_1(k)\\\vdots\\\theta_n(k)
      \end{bmatrix},}
    \end{split}
  \end{equation}
  where
  \begin{equation}\label{eqn:H}
    T =[T_1,T_2,\cdots,T_m], \; H= [H_1,H_2,\cdots,H_n].
  \end{equation}
  It holds that system \eqref{eqn:aft_mdl_reduce} shares the same transfer function with \eqref{eqn:fullmodel}.
\end{theorem}

\begin{proof}
The proof is presented in Appendix-\ref{sec:app_mdl_reduce}.
\end{proof}

\revise{Therefore, by performing model reduction, we present system \eqref{eqn:aft_mdl_reduce} which shares the same transfer function with \eqref{eqn:fullmodel} but with a reduced order. As proved previously, the output of \eqref{eqn:fullmodel} is the optimal Kalman estimate. As a result, \eqref{eqn:aft_mdl_reduce} also has the Kalman estimate as its output and the Kalman filter can be perfectly recovered by \eqref{eqn:aft_mdl_reduce} as well.}
We hereby refer both \eqref{eqn:fullmodel} and \eqref{eqn:aft_mdl_reduce} to lossless decomposition of Kalman fiter. Depending on the size of $m$ and $n$, one should use a system with smaller dimension to represent the centralized Kalman filter.  

\section{Local Implementation of Kalman filter}\label{sec:algorithm}
From Fig.~\ref{fig:infoflow}, it is clear that local decomposition proposed in Section~\ref{sec:decompose} is still centralized as a fusion center is required for calculating the weighted sum. In this section, we shall provide distributed algorithms for implementing it, where each sensor node performs local filtering by using the results from Section \ref{sec:decompose}, and global fusion by exchanging information with neighbors and running consensus algorithm. Based on whether $n$ is greater than $m$ or not, different algorithms will be presented to achieve a low communication complexity.

We use a weighted undirected graph  $\mathcal{G}=\{\mathcal{V},\mathcal{E},\mathcal{A}\}$ to model the interaction among nodes, where $\mathcal{V} =\{1,2,...,m\}$ is the set of sensors, $\mathcal{E}\subset \mathcal{V}\times\mathcal{V}$ is the set of edges, and $\mathcal{A}=\left[a_{i j}\right]$ is the weighted adjacency matrix. It is assumed $a_{ij}\geq 0$ and $a_{ij}=a_{ji},\forall i,j \in \mathcal{V}$. An edge between sensors $i$ and $j$ is denoted by $e_{ij}\in \mathcal{E}$, indicating that these two agents can communicate directly with each other. Note that $e_{ij}\in \mathcal{E}$ if and only $a_{ij}>0$. By denoting the degree matrix as $\mathcal{D} \triangleq \diag\left(\operatorname{deg}_{1}, \ldots, \operatorname{deg}_{N}\right)$ with $\mathrm{deg}_{i}=\sum_{j=1}^{N} a_{ij},$ the Laplacian matrix of $\mathcal{G}$ is defined as $\mathcal{L}_{\mathcal{G}}\triangleq\mathcal{D}-\mathcal{A}$. In this paper, a connected network is considered. We therefore can arrange the eigenvalues of Laplacian matrix as $0=\mu_1< \mu_2 \leq \cdots \leq \mu_m.$


\subsection{Description of the distributed estimator}\label{sec:analysis}
\revise{In light of \eqref{eqn:localdecompose}, the optimal estimate fuses $\hat{\xi}_i(k)$ from all sensors. However, in a distributed framework, each sensor can only access the information in its neighborhood. Hence, any sensor $i$ needs to, through the communication over network, infer $\hat{\xi}_j(k)$ for all $j\in\mathcal{V}$ to achieve a stable local estimate. }
	
\revise{Let us denote by $\eta_{i, j}(k)$ as the inference from sensor $i$ on sensor $j$.
As will be proved later in this section, $\eta_{i,j}(k)$, by running a synchronization algorithm, can track $\frac{1}{m}\hat{\xi}_j(k)$ with bounded error. Hence, every sensor $i$ can make a decent inference on $\hat{\xi}_j(k)$.} 
	
\revise{By collecting its inference on all sensors together, each sensor $i$ keeps a local state as below:
	\begin{equation}
		\eta_i(k) \triangleq
		{\left[\begin{array}{c}
				\eta_{i, 1}(k) \\
				\vdots \\
				\eta_{i, m}(k)
			\end{array}\right]\in\mathbb{R}^{mn},} 
	\end{equation}
 which will be updated by synchronization algorithms. Since $\eta_i(k)$ contains the fair inference on all $\hat{\xi}_j(k), j\in\mathcal{V}$, sensor $i$ finally uses it to compute a stable local estimate.}

To be concrete, let us define the message sent by agent $i$ at time $k$ as $\Delta_i(k)\triangleq\tilde{\Gamma}\eta_i(k)\in\mathbb{R}^m$, where $\tilde{\Gamma}=I_m\otimes\Gamma$ and $\Gamma$ is a design parameter to be given later. We are now ready to present the main algorithm. Suppose each node $i$ is initialized with $\hat{x}_i(0)=0$ and $\eta_i(0)=0$.  At any instant $k>0$, its update is outlined in Algorithm \ref{alg:unobservable}, the information flow of which is shown in Fig. \ref{fig:blkdiag}. Compared with Fig.~\ref{fig:infoflow}, the proposed algorithm is achieved in a distributed manner.

\revise{\begin{remark}
Instead of transmitting the raw estimate $\eta_{i}(k)\in\mathbb{R}^{mn}$, each agent sends a ``coded" vector $\Delta_i(k)$, with a smaller size $m$.
\end{remark}
}

\begin{algorithm}
  1:\: Using the latest measurement from itself, sensor $i$ computes the local residual and update the local estimate by
  \begin{equation}\label{eqn:step1}
  \begin{split}
  z_i(k) = y_i(k+1)-\beta^T\hat{\xi}_i(k),\\
  \hat{\xi}_i(k+1) = S \hat{\xi}_i(k)+\1_n z_i(k).
  \end{split}
  \end{equation}
  2:\: Compute $\Delta_i(k)=\tilde{\Gamma}\eta_i(k)$ and collect $\Delta_j(k)$ from neighbors and fuse the neighboring information with the consensus algorithm as
  \begin{equation}\label{eqn:eta}
    \eta_i(k+1) =\tilde{S}\eta_i(k) +\tilde{L}_iz_i(k)+\tilde{B}\sum_{j=1}^m a_{ij}(\Delta_j(k)-\Delta_i(k)),
  \end{equation}
  where $\tilde{S}$ and $\tilde{L}_i$ are given in \eqref{eqn:SL}, and $\tilde{B} \triangleq I_m \otimes \1_n$.\\
  3:\: Update the fused estimate on system state as:
  \begin{equation}\label{eqn:localfuse}
    \breve{x}_i(k+1)=mF\eta_i(k+1).
  \end{equation}
  4:\: Transmit the new state $\Delta_i(k+1)$ to neighbors.
  \caption{Distributed estimation algorithm for sensor $i$}
  \label{alg:unobservable}
\end{algorithm}

\begin{figure}
  \centering
  \resizebox{0.33\textwidth}{!}{\begin{tikzpicture}[auto, node distance=1.8cm,>=latex']

\node at (0,3.5) {$y_i(k)$};
\node at (4,3.5) {$y_j(k)$};

\draw [->] (0,3.3) -- (0,2.7);
\draw [->] (4,3.3) -- (4,2.7);

\node [sensor,align=center] (est1) {Linear system\\$\eta_{i}(k)$};
\node [sensor, right of=est1,node distance=4cm,align=center] (est2) {Linear system\\$\eta_{j}(k)$};
\node [est,above of=est1,align=center,node distance=2.2cm] (sensor1) {Local filter\\$\hat{\xi}_i(k)$};
\node [est,above of=est2,align=center,node distance=2.2cm] (sensor2) {Local filter\\$\hat{\xi}_j(k)$};

\draw [->] (0,1.7) --  (0,0.45);
\draw [->] (4,1.7) -- (4,0.45);

\node at (0.5,1.4) {$z_i(k)$};
\node at (4.5,1.4) {$z_j(k)$};

\draw [->] (est1) -- node {$\Delta_i(k)$} (est2);
\draw [->] (est2) -- node {$\Delta_j(k)$} (est1);

\draw[dashed, line width=0.3mm] (-1.5, 1.15) -- (5.5, 1.15);
\draw[dashed, line width=0.3mm] (-1.5, -0.8) -- (5.5, -0.8);

\node at (2,0.9) {Synchronization};

\draw [->] (est1) --  (0,-1.2);
\draw [->] (est2) -- (4,-1.2);

\node at (0,-1.4) {$\breve{x}_i(k)$};
\node at (4,-1.4) {$\breve{x}_j(k)$};
\end{tikzpicture}}
  \caption{The information flow of Algorithm \ref{alg:unobservable}, where nodes $i$ and $j$ are immediate neighbors.}
  \label{fig:blkdiag}
\end{figure}
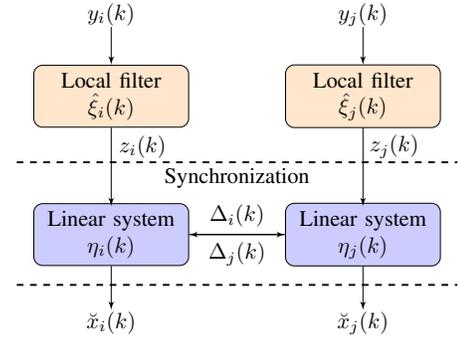

\subsection{Performance analysis}\label{sec:performance}
This part is devoted to the performance analysis of Algorithm \ref{alg:unobservable}. We shall first provide the following theorem: 
\begin{theorem}\label{thm:optimal}
	With Algorithm \ref{alg:unobservable}, the average of fused estimates from all sensors coincides with the optimal Kalman estimate at any instant $k$. That is,
	\begin{equation}\label{eqn:ave}
	\frac{1}{m}\sum_{i=1}^m \breve{x}_i(k)=\hat{x}(k), \forall k\geq 0.
	\end{equation}
\end{theorem}
\begin{proof}
	Summing \eqref{eqn:eta} over all $i=1,2,...,m$ yields
	\begin{equation}\label{eqn:etasum}
	\sum_{i=1}^m \eta_i(k+1) =\tilde{S}\sum_{i=1}^m\eta_i(k) +\sum_{i=1}^m\tilde{L}_i z_i(k), 
	\end{equation}
	where we use the fact that $a_{ij}=a_{ji}$ for any $i,j\in\mathcal{V}$. Comparing it with \eqref{eqn:xi_z}, it holds for any instant $k$ and any $j\in\mathcal{V}$ that:
	\begin{equation}\label{eqn:xvseta}
		\hat{\xi}_j(k) = \sum_{i=1}^m \eta_{i, j}(k).
	\end{equation}
	Therefore, the following equation is satisfied at any $k\geq 0$:
	\begin{equation}
	\begin{split}
	\frac{1}{m}\sum_{i=1}^m \breve{x}_i(k)&=\sum_{i=1}^m F\eta_i(k)=\sum_{i=1}^m\sum_{j=1}^m F_j\eta_{i, j}(k)\\&=\sum_{j=1}^m F_j \Big[\sum_{i=1}^m \eta_{i, j}(k)\Big]=\sum_{j=1}^m F_j\hat{\xi}_j(k)=\hat{x}(k).
	\end{split}
	\end{equation}
	This completes the proof.
\end{proof}

On the other hand, in order to show the stability of proposed estimator, it is also desired to prove the boundedness of error covariance. Towards this end, we introduce the following lemma, the condition of which is characterized in terms of a certain relation between the Mahler measure (the absolute product of unstable eigenvalues of $S$) and the graph condition number (the ratio of the maximum and minimum nonzero eigenvalues of the Laplacian matrix):

\begin{lemma}\label{lmm:poleplacement}
  Suppose that the product of all unstable eigenvalues of matrix $S$ meets the following condition:
  \begin{equation}\label{eqn:unstable}
    \prod_j |\lambda_j^u(S)| < \frac{1+\mu_2/\mu_m}{1-\mu_2/\mu_m},
  \end{equation}
  where $\lambda_j^u(S)$ represents the $j$th unstable eigenvalue of $S$.
  Let 
  \begin{equation}\label{eqn:Gamma}
    \Gamma=\frac{2}{\mu_2+\mu_m}\frac{\1_n^T\mathcal{P}S}{\1_n^T\mathcal{P}\1_n}\revise{\in\mathbb{R}^{1\times n}},
  \end{equation}
  where \revise{$\mu_2$
  	and $\mu_m$ are, respectively, the second smallest and largest
  	eigenvalues of  $\mathcal L_\mathcal{G}$.} Moreover, $\mathcal{P}>0$ is the solution to the following modified algebraic Riccati inequality
  \begin{equation}\label{eqn:riccati}
  	\mathcal{P}-S^{T} \mathcal{P} S+\left(1-\zeta^{2}\right) \frac{S^{T} \mathcal{P} \1_n \1_n^{T} \mathcal{P} S}{\1_n^{T} \mathcal{P} \1_n}>0,
  \end{equation}
  with $\zeta$ satisfying $\prod_{j}\left|\lambda_{j}^{u}(S)\right|<\zeta^{-1} \leq\frac{1+\mu_{2} / \mu_{m}}{1-\mu_{2} / \mu_{m}}.$ Then for any $j\in\{2,...,n\}$, it holds that 
  \begin{equation}\label{eqn:Seig}
  \rho(S-\mu_j \1_n \Gamma)<1.
  \end{equation}
\end{lemma}
\begin{proof}
For any $j\in\{2,...,n\}$, let us denote $\zeta_j = 1-2\mu_j/(\mu_2+\mu_m) \leq \zeta$. Since $(S,\1_n)$ is controllable, there exists some $\mathcal{P}>0$ which solves \eqref{eqn:riccati}. Together with \eqref{eqn:Gamma}, it holds that
\begin{equation}
	\begin{split}
		&(S-\mu_j \1_n \Gamma)^T \mathcal{P} (S-\mu_j \1_n \Gamma)-\mathcal{P}\\
		= &S^T\mathcal{P}S-(1-\zeta_j^2)\frac{S^{T} \mathcal{P} \1_n \1_n^{T} \mathcal{P} S}{\1_n^{T} \mathcal{P} \1_n}-\mathcal{P}\\
		\leq & S^T\mathcal{P}S-(1-\zeta^2)\frac{S^{T} \mathcal{P} \1_n \1_n^{T} \mathcal{P} S}{\1_n^{T} \mathcal{P} \1_n}-\mathcal{P}<0.
	\end{split}
\end{equation}
Hence, our proof completes.
\end{proof}

\begin{remark}
Note that, if all the eigenvalues of $S$ lie on or outside the unit circle, You \text{et al.} \cite{you2011network} prove that \eqref{eqn:Seig} holds if and only if \revise{\eqref{eqn:unstable}} is satisfied. In Lemma~\ref{lmm:poleplacement}, we further show that, \revise{\eqref{eqn:unstable}} is still a sufficient condition to facilitate \eqref{eqn:Seig} if $S$ has stable modes.
\end{remark}

\begin{remark}
Invoking Remark \ref{rmk:eigenvalue}, each $\lambda_j^u(S)$ corresponds to a root of the characteristic polynomial of $A^u$. Thus, the condition \revise{\eqref{eqn:unstable}} can be rewritten using the system matrix $A^u$,  
  \begin{equation}
    \prod_j |r_j(A^u)| < \frac{1+\mu_2/\mu_m}{1-\mu_2/\mu_m},
  \end{equation}
where $r_j(A^u$ is a root of the characteristic polynomial of $A^u$.
\end{remark}
With the above preparations, we are now ready to analyze the error covariance of local estimator as below:

\begin{theorem}\label{thm:observable}
Suppose that the Mahler measure of $S$ meets condition \revise{\eqref{eqn:unstable}}, and $\Gamma$ is designed based on \eqref{eqn:Gamma}--\eqref{eqn:riccati}.
With Algorithm \ref{alg:unobservable}, the error covariance of each local estimate $\breve{x}_i(k)$ is bounded at any instant $k$. 
\end{theorem}
\begin{proof}
Due to space limitation, the proof is given in Appendix-\ref{sec:appB}. 
  \end{proof}

The proof of Theorem \ref{thm:observable} implies that we present a distributed estimation scheme with quantifiable performance.  
\revise{
\begin{corollary}\label{col:error}
Suppose that the Mahler measure of $S$ meets condition \revise{\eqref{eqn:unstable}}, and $\Gamma$ is designed based on \eqref{eqn:Gamma}--\eqref{eqn:riccati}. Let $\breve W$ be the asymptotic error covariance of local estimates. Namely,
$$\breve W\triangleq \lim_{k\to\infty} \cov(\breve{e}(k)),$$  
where $\breve{e}(k)\triangleq \col[(\breve{x}_1(k)-x(k)),\cdots,(\breve{x}_m(k)-x(k))]$.
By using Algorithm \ref{alg:unobservable}, it holds that
 \begin{equation}
	\breve W = \bar W + (\1_m\1_m^T )\otimes P,
\end{equation}
where $\bar W$ is the asymptotic error covariance between local estimate and the Kalman estimate, and $P$ is the error covariance of Kalman filter as defined in \eqref{eqn:KFcov}. 
Moreover, $\breve W$ can be exactly calculated.
\end{corollary}
}

 \revise{As seen from the calculation, $\bar W$, i.e., the performance gap between our estimator and the optimal Kalman filter, is purely caused by the consensus error.} Therefore, if infinite consensus steps are allowed between two consecutive sampling instants, the consensus error vanishes and the performance of the proposed estimator coincides with that of the Kalman filter. 

Combining Theorems \ref{thm:optimal} and \ref{thm:observable}, the local estimator is stable at each sensor side. Therefore, we conclude that by applying the algorithm designed for linear system synchronization, i.e., \eqref{eqn:eta}, the problem of distributed state estimation is resolved.

\begin{remark}\label{rmk:n<m}
	Note that Algorithm \ref{alg:unobservable} requires each agent to send out an $m$-dimensional vector $\Delta_i(k)$ at any time. Therefore, in the network with a large number of sensors, i.e., $n<m$, this solution will cause a high communication cost. To address this issue, this remark, by leveraging the reduced-order estimator \eqref{eqn:aft_mdl_reduce} in Theorem \ref{lmm:mdl_reduce}, modifies Algorithm \ref{alg:unobservable} to introduce less communication complexity. To be specific, we aim to implement the reduced order system \eqref{eqn:aft_mdl_reduce} with distributed estimators. Similar as before, any agent $i$ stores its estimate on all the others in a variable $\vartheta_i(k)$, where 
	\begin{equation} 
	\vartheta_i(k) \triangleq
	{\left[\begin{array}{c}
		\vartheta_{i, 1}(k) \\
		\vdots \\
		\revise{\vartheta_{i, n}(k)}
		\end{array}\right]\in\mathbb{R}^{n^2}.} 
	\end{equation}
For each sensor $i$, it is initialized with $\hat{x}_i(0)=0$ and $\vartheta_i(0)=0$. For the case of $n<m$, the estimation algorithm works as in Algorithm~\ref{alg:n<m}. Following similar arguments, the local estimator at each sensor side is proved to be stable.

\revise{Combining it with Algorithm \ref{alg:unobservable}, we conclude the size of message sent by each sensor at any time is $\min\{m,n\}$. Compared with the existing solutions in distributed estimation, e.g., \cite{olfati2005distributed,olfati2009kalman,battistelli2014consensus,li2011consensus,battistelli2016stability}, our algorithm enjoys lower message complexity.}

\begin{remark}\label{rmk:replace}
	Notice that sensor node $i$ has perfect information of its own local estimate $\xi_i(k)$. Therefore, instead of using $\eta_{i, i}(k)$ to infer $\xi_i(k)/m$, node $i$ can just use $\xi_i(k)/m$ to replace $\eta_{i, i}(k)$ in \eqref{eqn:localfuse}, which potentially improves the performance of the estimators.  
\end{remark}


\end{remark}
  \begin{algorithm}
	1:\: Using the latest measurement from itself, sensor $i$ computes the local residual and update the local estimate by
	\begin{equation*}
		\begin{split}
			z_i(k) = y_i(k+1)-\beta^T\hat{\xi}_i(k),\\
			\hat{\xi}_i(k+1) = S \hat{\xi}_i(k)+\1_n z_i(k).
		\end{split}
	\end{equation*}
	2:\: Compute $\Delta_i(k)=(I_n\otimes\Gamma)\vartheta_i(k)$ such that $\Gamma$ is calculated by \eqref{eqn:Gamma}. Collect $\Delta_j(k)$ from neighbors and fuse the neighboring information with the consensus algorithm as
	\begin{equation}
		\begin{split}
			\vartheta_i(k+1) =&(I_n\otimes S)\vartheta_i(k) +T_iz_i(k)\\&+ (I_n \otimes \1_n)\sum_{j=1}^m a_{ij}(\Delta_j(k)-\Delta_i(k)),
		\end{split}
	\end{equation}
	where $T_i$ is defined in \eqref{eqn:T_i}.\\
	3:\: Update the fused estimate on system state as:
	\begin{equation}
		\breve{x}_i(k+1)=mH\vartheta_i(k+1),
	\end{equation}
	where $H$ is given in \eqref{eqn:H}.\\
	4:\: Transmit the new state $\Delta_i(k+1)$ to neighbors.
	\caption{Distributed estimation algorithm $2$ for sensor $i$}
	\label{alg:n<m}
\end{algorithm}

\section{Extensions of Proposed Solutions}\label{sec:discuss}
In the previous sections, we leverage the linear system synchronization algorithm proposed in \cite{you2011network}, to solve the problem of distributed state estimation. In this section, we aim to extend such a result and show that any control strategy, which can facilitate the linear system synchronization, can be modified to yield a stable distributed estimator. As a result, we bridge the fields of distributed state estimation and linear system synchronization. 

Let us consider the synchronization of the following homogeneous LTI system:
  \begin{align}\label{eqn:linear}
  	\eta_i(k+1) &= \tilde{S}\eta_i(k) + \tilde{B}u_i(k), \;\forall i\in\mathcal{V},
  \end{align}
where $u_i(k)$ is the control input of agent $i$. In literature, a large variety of synchronization algorithms has been proposed with the framework below:
\begin{equation}\label{eqn:general}
\begin{split}
	&\omega_i(k+1) = \mathcal{A}\omega_i(k)+\mathcal{B}\eta_i(k+1), \\
	&\Delta_i(k) = \tilde{\Gamma} \omega_i(k),\\
	&u_i(k) = \sum_{j=1}^m a_{ij}\gamma_{ij}(k)(\Delta_j(k)-\Delta_i(k)),  
\end{split}
\end{equation}
where $\omega_i(k)$ is the ``hidden state" that is necessary for agent $i$ to yield the communication state $\Delta_i(k)$ and input $u_i(k)$, and $\tilde{\Gamma}$ refers to the control gain. Notice that \eqref{eqn:general} can be used to model the controller with memory. Moreover, $\gamma_{ij}(k)\in [0,1]$ models the fading or lossy effect of the communication channel from agent $j$ to agent $i$. At every time, the agent collects the available information in its neighborhood and synthesizes its communication state and control signal via \eqref{eqn:general}. 


For simplicity, we denote $\mathcal{U}$ as the control strategy that can be represented by \eqref{eqn:general}. Let the average of local states at time $k$ be
\begin{equation*}
\bar{\eta}(k) = \frac{1}{m} \sum_{i=1}^m\eta_i(k).
\end{equation*}
The network of subsystems \eqref{eqn:linear} reaches \textit{strong} synchronization under $\mathcal{U}$, if the following statements hold at any time:

\begin{enumerate}
\item \textbf{Consistency:} the average of local states keeps consistent throughout the execution, i.e., 
\begin{equation}\label{eqn:consistency}
	\bar{\eta}(k+1) =\tilde{S}\bar{\eta}(k).
\end{equation}
\item \textbf{Exponential Stability:} agents exponentially reach consensus in mean square sense, i.e., there exist $c>0$ and $\rho\in(0,1)$ such that 
\begin{equation}\label{eqn:consensus}
	\mathbb{E}[||\eta_i(k)-\bar{\eta}(k)||^2] \leq c\rho^{k}, \;\forall i\in\mathcal{V}.
\end{equation}
\end{enumerate}

We now review several existing strategies which facilitate the strong synchronizationand show that they can be represented by \eqref{eqn:general}: 

\begin{enumerate}
\item Let $\Delta_i(k)=\tilde{\Gamma}\eta_i(k)$ be the communication state defined in Section~\ref{sec:analysis}. To facilitate the synchronization of homogeneous linear systems in undirected communication topology, You \textit{et al.} \cite{you2011network} design the following control law:
\begin{equation}\label{eqn:You}
	u_i(k) = \sum_{j=1}^m a_{ij}(\Delta_j(k)-\Delta_i(k)),
\end{equation} 
which coincides with \eqref{eqn:general}. 
\item Another example is the filtered consensus protocol given in \cite{gu2011consensusability}. By designing the hidden state as
\begin{equation}
	\omega_i(k) = F(q)\eta_i(k),
\end{equation}
where $q$ is the unit advance operator, i.e., $q^{-1}s(k)=s(k-1)$, and $F(z)$ is the transfer function of a square stable filter, the synchronization of linear systems is achieved by \eqref{eqn:general} under a more relaxed condition than \revise{\eqref{eqn:unstable}}, that is: 
$
\prod_j |\lambda_j^u(S)| < \frac{1+\sqrt{\mu_2/\mu_m}}{1-\sqrt{\mu_2/\mu_m}}.
$
\item Instead of focusing on perfect communication channels, the authors in \cite{you2013consensus} and \cite{xu2019distributed} develop the control protocols to account for the random failure on communication links and Markovian switching topologies, respectively. By modeling the packet loss with the Bernoulli random variable $\gamma_{ij}(k)\in\{0,1\}$, these works complement the results in \cite{you2011network} and prove the mean square stability under the control strategy \eqref{eqn:general}.
\end{enumerate}

Notice that Algorithms~\ref{alg:unobservable} and \ref{alg:n<m} utilize \eqref{eqn:You} for achieving synchronization and producing stable distributed estimators. In what follows, we argue that the optimal Kalman estimate can indeed be distributively implemented using any linear system synchronization algorithms facilitating \eqref{eqn:consistency}-\eqref{eqn:consensus}. To be specific, Algorithm~\ref{alg:unobservable} should be modified\footnote{Similarly, in the case of $n<m$, one can also derive the general form of Algorithm~\ref{alg:n<m} with any linear system synchronization strategy $\mathcal{U}$.} by replacing \eqref{eqn:eta} with
\begin{equation}\label{eqn:eta2}
\eta_i(k+1) =\tilde{S}\eta_i(k)+\tilde{B}u_i(k)+\tilde{L}_iz_i(k),
\end{equation}
 where $u_i(k)$ is generated by $\mathcal{U}$ that facilitates \eqref{eqn:consistency}-\eqref{eqn:consensus}. We then state the stability of local estimators as below:

\begin{theorem}\label{thm:general}
Consider any algorithm $\mathcal{U}$ which facilitates the statements \eqref{eqn:consistency} and \eqref{eqn:consensus}. At any time $k$, suppose each $\gamma_{ij}(k)$ is independent of the noise $\{w(k)\}$ and $\{v(k)\}$. Then \eqref{eqn:eta2} yields a stable estimator for each sensor node. Specifically, the following statements hold for any $k\geq 0$:
\begin{enumerate}
	\item the average of local estimates from all sensor coincides with the optimal Kalman estimate;
	\item the error covariance of each local estimate is bounded. 
\end{enumerate} 
\end{theorem}

\begin{proof}
The proof is given in Appendix-\ref{sec:appC}. 
\end{proof}

\begin{remark}
	Theorem~\ref{thm:general} assumes the independence of the communication topology and system/measurement noises. Therefore, as for the event-based synchronization algorithms, where the communication relies on the agents' states, we cannot analyze its efficiency of solving the distributed estimation problem by directly resorting to Theorem~\ref{thm:general}. In the future work, we will continue to investigate this topic.
\end{remark}


\revise{In contrast with Fig~\ref{fig:existing}, this work, by using the lossless decomposition of Kalman filter, decouples the local filter from the consensus process, as shown in Fig.~\ref{fig:blkdiag}. The decoupling enables us to leverage the rich results in linear systems synchronization to analyze the performance of local estimators, as proved in Theorem~\ref{thm:general}. Moreover, following the similar proof arguments as that of Theorem~\ref{thm:observable}, we can show that with our framework, the error covariance of each local estimate actually consists of two orthogonal parts: the inherent estimation error of Kalman filter and the distance from local estimate to Kalman filter, namely:
\begin{equation*}
	\begin{aligned}
	&\cov(\breve e_i(k))=\cov(\breve x_i(k)-x(k))\\=&\cov(\breve x_i(k)-\hat x(k)+\hat x(k)-x(k))\\
	=&\cov(\breve x_i(k)-\hat x(k))+\cov(\hat x(k)-x(k))\\
	=&\cov\Big(\breve x_i(k)-\frac{1}{m}\sum_{i=1}^m \breve x_i(k)\Big)+\cov(\hat x(k)-x(k))\\
	=&m^2F\cov\Big(\eta_i(k)-\frac{1}{m}\sum_{i=1}^m \eta_i(k)\Big)F^T+\cov(\hat x(k)-x(k)),
\end{aligned}
\end{equation*}
	where the third equality holds due to the optimality of Kalman filter, and the last equality holds by \eqref{eqn:localfuse}. Notice that the second term of RHS is the error covariance of Kalman filter, while first term is the error between local estimate and Kalman filter and purely determined by the consensus process. Therefore, by choosing proper strategy $\mathcal{U}$, extensive results on achieving strong synchronization can be applied to \eqref{eqn:eta2} to deal with the consensus error in various settings, such as directed graph, time-varying topologies, \textit{etc}. \modify{Particularly, if infinite consensus steps are allowed between two consecutive sampling instants, the consensus error vanishes, i.e., $\eta_i(k)-\frac{1}{m}\sum_{i=1}^m \eta_i(k)=0$, and the performance of the proposed estimator is optimal since it coincides with that of the Kalman filter.} That means the global optimality can be guaranteed.
}

  \section{Numerical Example}\label{sec:simulation}
  In this section, we present numerical examples to verify the theoretical results obtained in previous sections. 
 
 \subsection{Numerical example when $n<m$}
Let us consider the case where four sensors cooperatively estimate the system state.
  The system parameters are listed below:
  \begin{equation}
\begin{split}
& A = 
\begin{bmatrix}
0.9 & 0\\
0 & 1.1
\end{bmatrix},\;
C = 
\begin{bmatrix}
1 & 0 & 1 & 1\\
0 & 1 & 1 & -1
\end{bmatrix}^T,\\
&Q=0.25I_2, \;R=4I_4.
\end{split}
\end{equation}
\revise{In this example, the number of states is smaller than that of sensors, i.e. $n<m$. We therefore choose Algorithm~\ref{alg:n<m}.} Moreover, notice that the system is unstable, and sensor $1$ cannot observe the unstable state.

Suppose that the topology of these four sensors is a ring with weight $1$ for each edge. The Laplacian matrix is thus:
\begin{equation}
  \mathcal{L_G}=\begin{bmatrix}
    2 & -1 & 0 & -1\\
    -1 & 2 & -1 & 0\\
    0 & -1 & 2 & -1\\
    -1& 0 & -1 & 2
  \end{bmatrix}.
\end{equation}
\revise{It is not difficult to check that the second smallest and the largest eigenvalues of $\mathcal{L_G}$ are respectively $\mu_2=2$, $\mu_4=4$. To fulfill the sufficient condition in Lemma~\ref{lmm:poleplacement}, let us choose $\zeta=0.5$.}

%

We set the initial state $x(0)\sim \mathcal{N}(0,I)$ and the initial local estimate $\breve x_i(0)=0$ for each sensor $i\in\{1,2,3,4\}$. It can be seen that the mean squared local estimate error $e_i(k)$ enters steady state and is stable after a few steps (see Fig.~\ref{fig:x}).



\begin{figure}[htbp]
	\begin{minipage}[b]{0.5\textwidth}
		\centering
		\resizebox{0.9\textwidth}{!}{\input{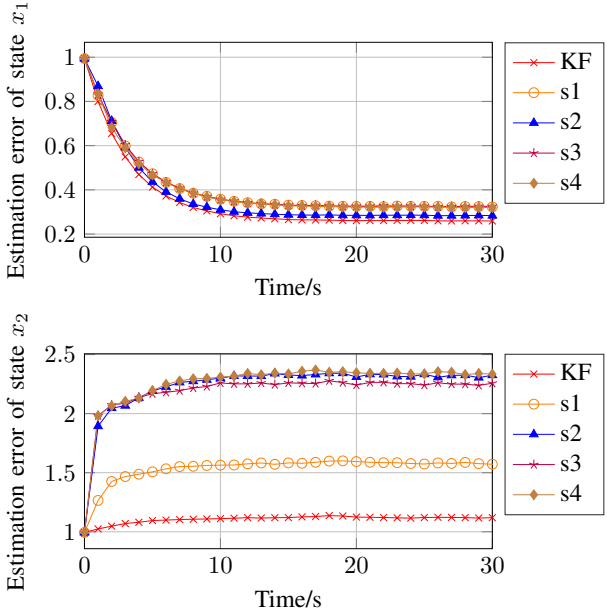}}
	\end{minipage}%
	\\
	\begin{minipage}[b]{0.5\textwidth}
		\centering
		\resizebox{0.9\textwidth}{!}{\begin{tikzpicture}
\begin{axis}[%
width=2.3in,
height=1.1in,
scale only axis,
xmin=0,
xmax=30,
xlabel={Time/s},
ylabel={Estimation error of state $x_2$},
axis background/.style={fill=white},
xmajorgrids,
ymajorgrids,
legend style={legend cell align=left, align=left, draw=white!15!black},
legend pos={outer north east}
]
    \addplot[color={red}, mark={x}]
        coordinates {
            (0,0.9943765995349962)
            (1,1.0241254689704056)
            (2,1.0473887499847159)
            (3,1.0697339448807577)
            (4,1.0803713746007921)
            (5,1.0949033915033333)
            (6,1.0991482380149056)
            (7,1.1034194862558937)
            (8,1.106307793947288)
            (9,1.1085925380380288)
            (10,1.110999599114129)
            (11,1.1147645806759343)
            (12,1.1193703731575648)
            (13,1.116094933319016)
            (14,1.1192249437084907)
            (15,1.1199115573191498)
            (16,1.1260096469155687)
            (17,1.1291450572405135)
            (18,1.136950337574846)
            (19,1.133455361109355)
            (20,1.1244685761514959)
            (21,1.1227919434732434)
            (22,1.1206090843375092)
            (23,1.1180481048662483)
            (24,1.1155404791626586)
            (25,1.1204428581200438)
            (26,1.1224462151114232)
            (27,1.121752685683367)
            (28,1.1196943320049257)
            (29,1.1164613963806695)
            (30,1.1193740859610601)
            (31,1.1196678829620261)
            (32,1.1262821415531057)
            (33,1.1156317612924114)
            (34,1.116567512625055)
            (35,1.116742559368047)
            (36,1.117315577996865)
            (37,1.118028323124454)
            (38,1.1164678892437345)
            (39,1.117983506365516)
            (40,1.123098865189546)
            (41,1.1230946876927135)
            (42,1.1225522521222113)
            (43,1.1224618836995262)
            (44,1.1238984888915637)
            (45,1.125901217233803)
            (46,1.1226315332791659)
            (47,1.119965687831627)
            (48,1.1193711594564655)
            (49,1.1226049410413392)
            (50,1.1241068397421996)
        }
        ;
    \addlegendentry {KF}
    \addplot[color={orange}, mark={o}]
        coordinates {
            (0,0.9943765995349962)
            (1,1.2652924039335356)
            (2,1.4256320040837056)
            (3,1.4673510411562172)
            (4,1.4873395687011577)
            (5,1.5073301774920105)
            (6,1.5329161641112685)
            (7,1.5520030589757723)
            (8,1.554391863843294)
            (9,1.5615200183722036)
            (10,1.563923717117754)
            (11,1.5653870801152323)
            (12,1.574879924241972)
            (13,1.5825303535941295)
            (14,1.5711163184525245)
            (15,1.5831105788545525)
            (16,1.5799868657614462)
            (17,1.5882189783567848)
            (18,1.5979613761646876)
            (19,1.6003075631065073)
            (20,1.5941329592447682)
            (21,1.588284581504355)
            (22,1.5851274627279721)
            (23,1.584373968832322)
            (24,1.577858666299255)
            (25,1.5735674219348068)
            (26,1.5827615472696204)
            (27,1.5786961819804337)
            (28,1.587122886228931)
            (29,1.5765449008480295)
            (30,1.571758603170185)
            (31,1.5722071728043354)
            (32,1.5790341798123353)
            (33,1.5937749764613511)
            (34,1.5745873058034463)
            (35,1.5701470146906793)
            (36,1.5669255079112894)
            (37,1.5699117333623531)
            (38,1.57213903822218)
            (39,1.5750483423635424)
            (40,1.5761992587434166)
            (41,1.5812174368332743)
            (42,1.5817244605547665)
            (43,1.586364793112144)
            (44,1.5853213582576369)
            (45,1.5866770960927927)
            (46,1.5821300768923279)
            (47,1.5851061919445362)
            (48,1.5826051452764909)
            (49,1.572164251666703)
            (50,1.5822951282611506)
        }
        ;
    \addlegendentry {s1}
    \addplot[color={blue}, mark={triangle*}]
        coordinates {
            (0,0.9943765995349962)
            (1,1.8936625444249422)
            (2,2.0465205257360233)
            (3,2.066951276012711)
            (4,2.1338615693851746)
            (5,2.1922909740541803)
            (6,2.226487115248709)
            (7,2.263933580097268)
            (8,2.2742324432017647)
            (9,2.2866612655231235)
            (10,2.300561709236413)
            (11,2.3187874391349936)
            (12,2.3183107436885133)
            (13,2.3172153186124294)
            (14,2.339358433536219)
            (15,2.326601894094998)
            (16,2.322100645625816)
            (17,2.3311028833616803)
            (18,2.3412621523512342)
            (19,2.342196254698138)
            (20,2.309811332436182)
            (21,2.3326324701167347)
            (22,2.331147100728079)
            (23,2.314863629664187)
            (24,2.31285475802809)
            (25,2.3332171428385466)
            (26,2.3071309244263793)
            (27,2.3241146984247614)
            (28,2.3254061798842365)
            (29,2.304722589550026)
            (30,2.3250475757893168)
            (31,2.326875039249452)
            (32,2.347786042018165)
            (33,2.3050612298653994)
            (34,2.3143523106351727)
            (35,2.323050036506524)
            (36,2.3268902505681694)
            (37,2.326027817270796)
            (38,2.3264477462757323)
            (39,2.3118714464205596)
            (40,2.329829962182898)
            (41,2.313641678866345)
            (42,2.334185411755957)
            (43,2.3313963439610914)
            (44,2.319456118549934)
            (45,2.3358929879940047)
            (46,2.3168383584104326)
            (47,2.321115857214799)
            (48,2.3270072795412897)
            (49,2.3384601141712515)
            (50,2.318873680806025)
        }
        ;
    \addlegendentry {s2}
    \addplot[color={purple}, mark={star}]
        coordinates {
            (0,0.9943765995349962)
            (1,1.9828858762138304)
            (2,2.072953590569488)
            (3,2.0880223184728335)
            (4,2.1360134385063105)
            (5,2.168877836105569)
            (6,2.180518305551039)
            (7,2.194470270203278)
            (8,2.2170095545654873)
            (9,2.2294796048964476)
            (10,2.259813698050696)
            (11,2.252040729684584)
            (12,2.2518134092075446)
            (13,2.258969773117035)
            (14,2.2456372378476983)
            (15,2.260074704840985)
            (16,2.25690801887656)
            (17,2.2546645079335232)
            (18,2.2780084217125545)
            (19,2.2627407160130613)
            (20,2.241840788498835)
            (21,2.2639896937363018)
            (22,2.2653910716270422)
            (23,2.253397842090359)
            (24,2.253664694284399)
            (25,2.2409423831703603)
            (26,2.260721382240181)
            (27,2.2498742978256083)
            (28,2.249337306194815)
            (29,2.2387064936893575)
            (30,2.2572354436442903)
            (31,2.2559199593761075)
            (32,2.2629454575339145)
            (33,2.237792912635502)
            (34,2.259809633090449)
            (35,2.2600701221241586)
            (36,2.2590233377061026)
            (37,2.2566562394095957)
            (38,2.259200706317603)
            (39,2.2522596098668717)
            (40,2.271687711433266)
            (41,2.251367671607407)
            (42,2.2467096704728404)
            (43,2.2653285924423328)
            (44,2.249874385082106)
            (45,2.2693838984464265)
            (46,2.2443897174131746)
            (47,2.2700278666501776)
            (48,2.2596714046786364)
            (49,2.2710007617069685)
            (50,2.2585889284508913)
        }
        ;
    \addlegendentry {s3}
    \addplot[color={brown}, mark={diamond*}]
        coordinates {
            (0,0.9943765995349962)
            (1,1.9816540696459648)
            (2,2.0654615474192903)
            (3,2.1035165609535103)
            (4,2.1315441874490997)
            (5,2.194405042764902)
            (6,2.2477131762354583)
            (7,2.2761788789650046)
            (8,2.2947881679581665)
            (9,2.2977677725464183)
            (10,2.308375957467222)
            (11,2.317980480162738)
            (12,2.339427607324274)
            (13,2.32883551380075)
            (14,2.3427375592794983)
            (15,2.3356439111894307)
            (16,2.3585557228897294)
            (17,2.3690788525616533)
            (18,2.3493756113388495)
            (19,2.354762160569893)
            (20,2.3456192480187106)
            (21,2.3403966321669003)
            (22,2.340282752809893)
            (23,2.3432582094375234)
            (24,2.3354632251712575)
            (25,2.3358984300908086)
            (26,2.3548376694214577)
            (27,2.350751598558515)
            (28,2.330001313064622)
            (29,2.339936620551301)
            (30,2.334139579826345)
            (31,2.3296671247918894)
            (32,2.3477739816701084)
            (33,2.3372527594188757)
            (34,2.3294471328167496)
            (35,2.331507818950099)
            (36,2.3364833657462096)
            (37,2.3324300667271194)
            (38,2.3362024692199643)
            (39,2.3484213589008274)
            (40,2.332847142957334)
            (41,2.3589948008145787)
            (42,2.345722793829973)
            (43,2.3270312885418902)
            (44,2.3439433888447985)
            (45,2.3273308186147994)
            (46,2.36039655167833)
            (47,2.339409769673635)
            (48,2.336415886140528)
            (49,2.3442896822054715)
            (50,2.328126528073856)
        }
        ;
    \addlegendentry {s4}
\end{axis}
\end{tikzpicture}}
	\end{minipage}
	\caption{Average mean square estimation error of system states under Kalman filter and local estimators in 10000 experiments.}
	\label{fig:x}
\end{figure}

 \subsection{Numerical example when $n> m$}
In the second example, we simulate the heat transfer process \footnote{\revise{State estimation in diffusion process has wide applications in sensor network, e.g., urban CO$_2$ emission monitoring \cite{mao2012citysee}, temperature monitoring in data center \cite{parolini2011cyber}, etc. }} in a planar closed region discussed in \cite{Mo15HeatProcess} and \cite{MO2009174}:
\begin{equation}
  \frac{\partial u}{\partial t}=\alpha\Big(\frac{\partial^2u}{\partial x_1^2}+\frac{\partial^2u}{\partial x_2^2}\Big),
\end{equation}
with boundary conditions
\begin{equation}
  \frac{\partial u}{\partial x_1}\Big|_{t,0,x_2}=\frac{\partial u}{\partial x_1}\Big|_{t,l,x_2}=\frac{\partial u}{\partial x_2}\Big|_{t,x_1,0}=\frac{\partial u}{\partial x_2}\Big|_{t,x_1,l}=0,
\end{equation}
where $x_1$ and $x_2$ are the coordinates in the region; $u(t,x_1,x_2)$ indicates the temperature at time $t$ at position $(x_1,x_2)$, $l$ is the side length of the square region and $\alpha$ adjusts the speed of the diffusion process. With a $N\times N$ grid and sample frequency $1$Hz, the diffusion process can be discretized as:
\begin{equation}
  \label{equ:diffusion_process}
  \begin{split}
  u&(k+1,i,j)-u(k,i,j)=\frac{\alpha}{h^2}[u(k,i-1,j)+u(k,i,j-1)\\
  &+u(k,i+1,j)+u(k,i,j+1)-4u(k,i,j)],
  \end{split}
\end{equation}
where $h=\frac{l}{N-1}$ denotes the size of each grid and $u(k,i,j)$ indicates the temperature at time $k$ at location $(ih,jh)$. By collecting all the temperature values of each grid, we define the state variable $U(k)=[u(k,0,0),\cdots,u(k,0,N-1),u(k,1,0),\cdots,u(k,N-1,N-1))]^T$. Further, by introducing process noise into \eqref{equ:diffusion_process}, one derives the following system equation:
\begin{equation}
  U(k+1)=AU(k)+w(k),
\end{equation}
where $w(k)\sim \mathcal{N}(0,Q)$ is Gaussian noise.

As shown in Fig.~\ref{fig:heatmap}, $m$ sensors are randomly deployed in this region to monitor the temperature, where the measurement of each sensor is a linear combination of temperature of the grids around it. Specifically, suppose the location of sensor $s$ is $(\hat{x}_1,\hat{x}_2)$ such that $\hat{x}_1\in[i,i+1)$, $\hat{x}_2\in[j,j+1)$, we define $\Delta \hat{x}_1=x_{i1}-i$ and $\Delta \hat{x}_1=x_{i2}-j$. We assume that the measurement of sensor $s$ at time $k$ is
\begin{equation}
  \label{equ:diffusion_measurement}
  \begin{split}
  y_s(k)=&\frac{1}{h^2}\big[(1-\Delta \hat{x}_1)(1-\Delta \hat{x}_2)u(k,i,j)\\&+\Delta \hat{x}_1(1-\Delta \hat{x}_2)u(k,i+1,j)\\&+
  (1-\Delta \hat{x}_1)\Delta \hat{x}_2 u(k,i,j+1)\\
  &+\Delta \hat{x}_1\Delta \hat{x}_2u(k,i+1,j+1)\big]+v_s(k).
  \end{split}
\end{equation}
We collect the measurements of each sensor at time $k$ and denote it as $Y(k)$, then it follows
\begin{equation}
  Y(k)=CU(k)+v(k),
\end{equation}
where $v(k)\sim \mathcal{N}(0,R)$ is the measurement noise and the measurement matrix $C$ can be derived from \eqref{equ:diffusion_measurement}. The parameters for the simulation are listed below:
\begin{itemize}
  \item $\alpha=0.2$;
  \item $l=4$ and $N=5$, thus the grid size $h=1$;
  \item $n=N^2=25$ and $m=15$. \revise{Therefore, $n>m$, which is different from our first example.};
  \item $Q=0.2I_{25}$ and $R=3I_{15}$.
\end{itemize}

As discussed in Remark \ref{rmk:replace}, we replace $\eta_{i,i}(k)$ with the estimates given by local Kalman filters. The results are shown in Fig.~\ref{fig:heatmap}. Our algorithm achieves better performance compared with local Kalman filters which merely use the measurement of the sensor itself. The improvement of each sensor can be found in TABLE ~\ref{table:improvement}. Specifically, for each sensor $i$, we respectively define the performance of local Kalman filter and our algorithm in terms of:
\begin{equation}
  \varrho_{i1} \triangleq \frac{\tr(\hat{P}_i)}{\tr(P)},
  \varrho_{i2} \triangleq \frac{\tr(\breve{P}_i)}{\tr(P)},
\end{equation}
where $\hat{P}_i$, $\breve{P}_i$ and $P$ are respectively the steady-state error covariance of local Kalman filter, our estimator and centralized Kalman filter. We see that the proposed scheme outperforms the local Kalman filter by at least $50\%$ for each sensor.
\begin{figure}[]
	\centering
	\includegraphics[width=0.5\textwidth]{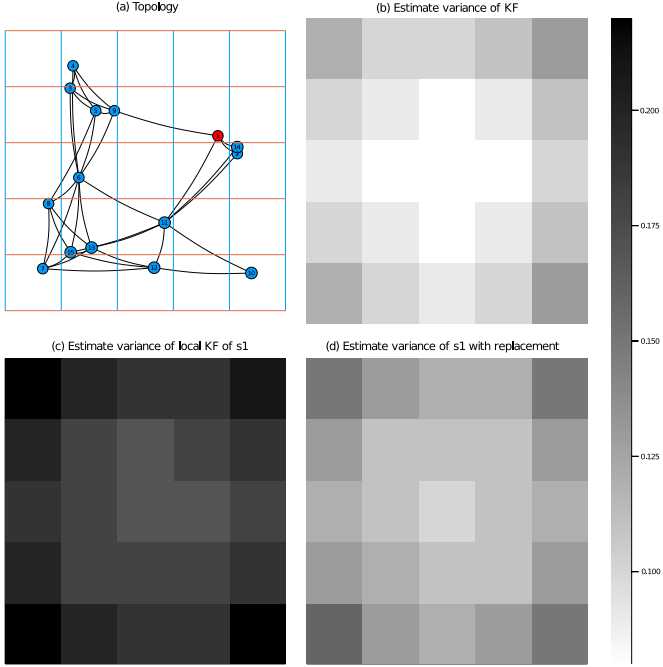}
	\caption{(a) The position and topology of $m$ sensors in the $N\times N$ grid lines; (b) The estimate variance of centralized Kalman filter; (c) The estimate variance of local Kalman filter; (d) The estimate variance of our estimators in 10000 experiments.}
	\label{fig:heatmap}
\end{figure}

\begin{table*}[]
	\small
	\caption{Performance Improvement in Comparison with Local Kalman filter}
	\label{table:improvement}
	\begin{tabular}{m{4.05cm}<{\centering} *{15}{m{0.49cm}<{\centering}}}
		\hline
		 Sensor index $i$
		& 1  &2 & 3 & 4 & 5 & 6 &7 & 8 & 9 & 10 & 11  & 12 & 13 & 14 & 15\\ \hline
		Local KF performance $\varrho_{i1}$ & 1.94 & 1.94 & 1.96 & 1.96 & 1.94 & 1.93 & 1.97 & 1.94 & 1.95 & 1.94& 1.92 & 1.94 & 1.95 & 1.94 & 1.95\\\hline
		Our estimator performance $\varrho_{i2}$ & 1.26 & 1.35 & 1.31 & 1.31 & 1.26 & 1.13 & 1.22 & 1.21 & 1.23 & 1.44& 1.12 & 1.22 & 1.18 & 1.35 & 1.18\\\hline
		Improvement $\varrho_{i1}-\varrho_{i2}$ & 68\%  &59\% & 65\% & 65\% & 68\% & 80\% & 75\% & 73\% & 71\% & 50\% & 80\%  & 72\% & 76\% & 59\% & 77\%  \\\hline
	\end{tabular}
\end{table*}

{\color{blue}\subsection{Comparison with existing algorithms}
We further compare the performance of Algorithm~\ref{alg:unobservable} with those of existing algorithms: 1) centralized Kalman filter (CKF), 2) KCF2009  (\hspace{1pt}\cite{olfati2009kalman}), and 3) CMKF2018 (\hspace{1pt}\cite{li2018weightedly}), through a numerical example on inverted pendulum. 

Notice that an inverted pendulum has $n=4$ states: $x=[p;\ \dot p;\ \theta;\ \dot\theta]$, namely, the cart position, cart velocity, pendulum angle from vertical and pendulum angle velocity, respectively.
We consider the system linearized at $\theta= \dot\theta = 0$ and discretized with sampling interval $T= 0.01s$, where the detailed system equation can be found in \cite{li2021efficient} with system noise $w(k)\sim \mathcal{N}(0,0.05^2I_n)$.

In the example, $m=4$ sensors are connected as a ring to infer the system state. Let the measurement equation be
\begin{equation}
	\begin{aligned}
		y(k)&=\left[\begin{array}{cccc}
			1 & 0 & 0 & 0 \\
			1 & 0 & 0 & 0 \\
			1 & 0 & 0 & 0 \\
			0 & 0 & 1 & 0
		\end{array}\right] x(k)+v(k),
	\end{aligned}
\end{equation}
where $v(k)\sim \mathcal{N}(0,0.3^2I_m)$. Notice that sensor $4$ cannot fully observe the state space. Fig.~\ref{fig:compare} illustrates the mean square error (MSE) of its estimate on $p$ and $\theta$, respectively. The results show that our algorithm yields better estimation performance.}

\begin{figure}
	\begin{minipage}{0.5\textwidth}
		\centering
		\includegraphics[width=0.8\textwidth]{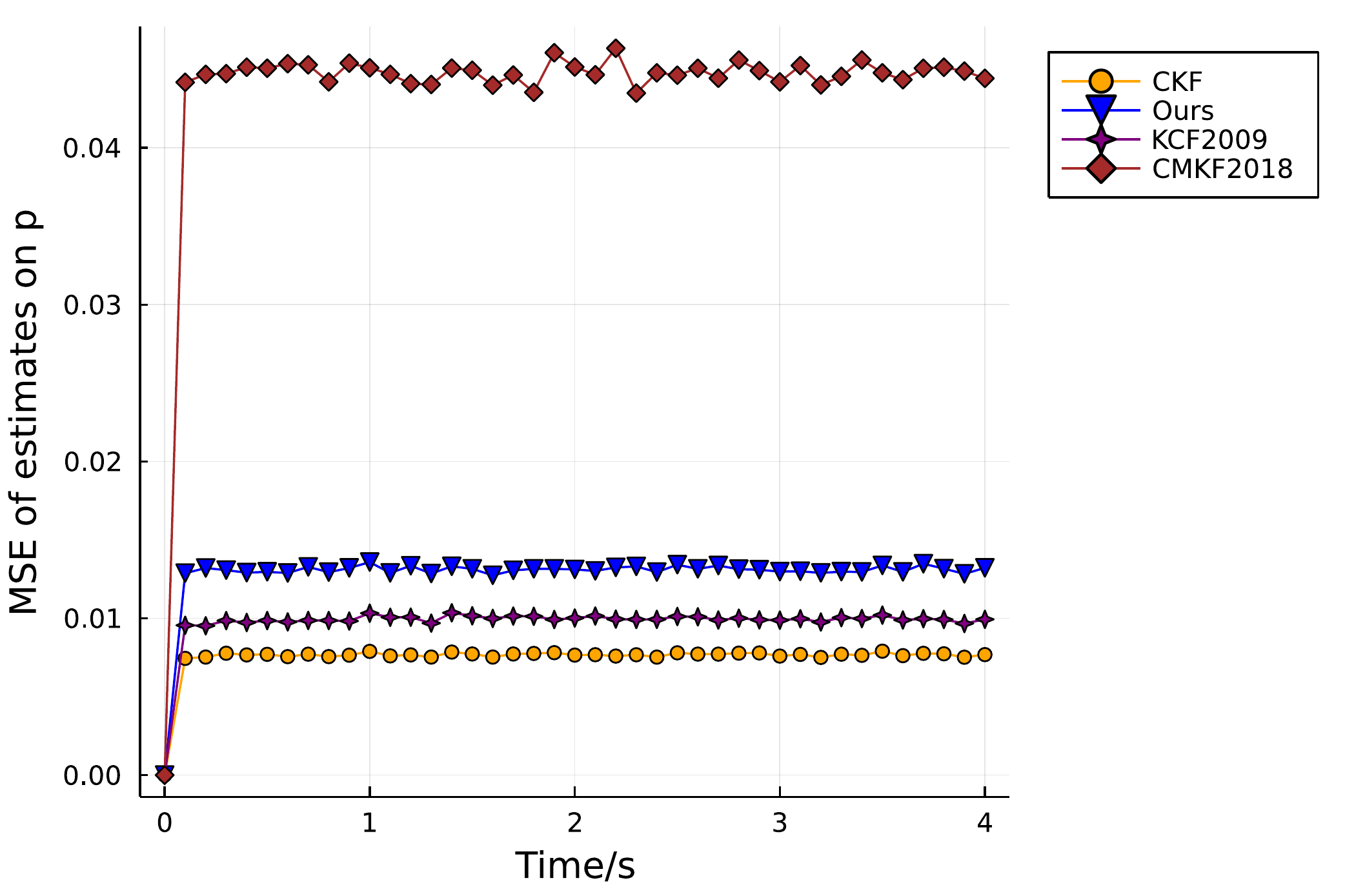}
	\end{minipage}
	\begin{minipage}{0.5\textwidth}
		\centering
		\includegraphics[width=0.8\textwidth]{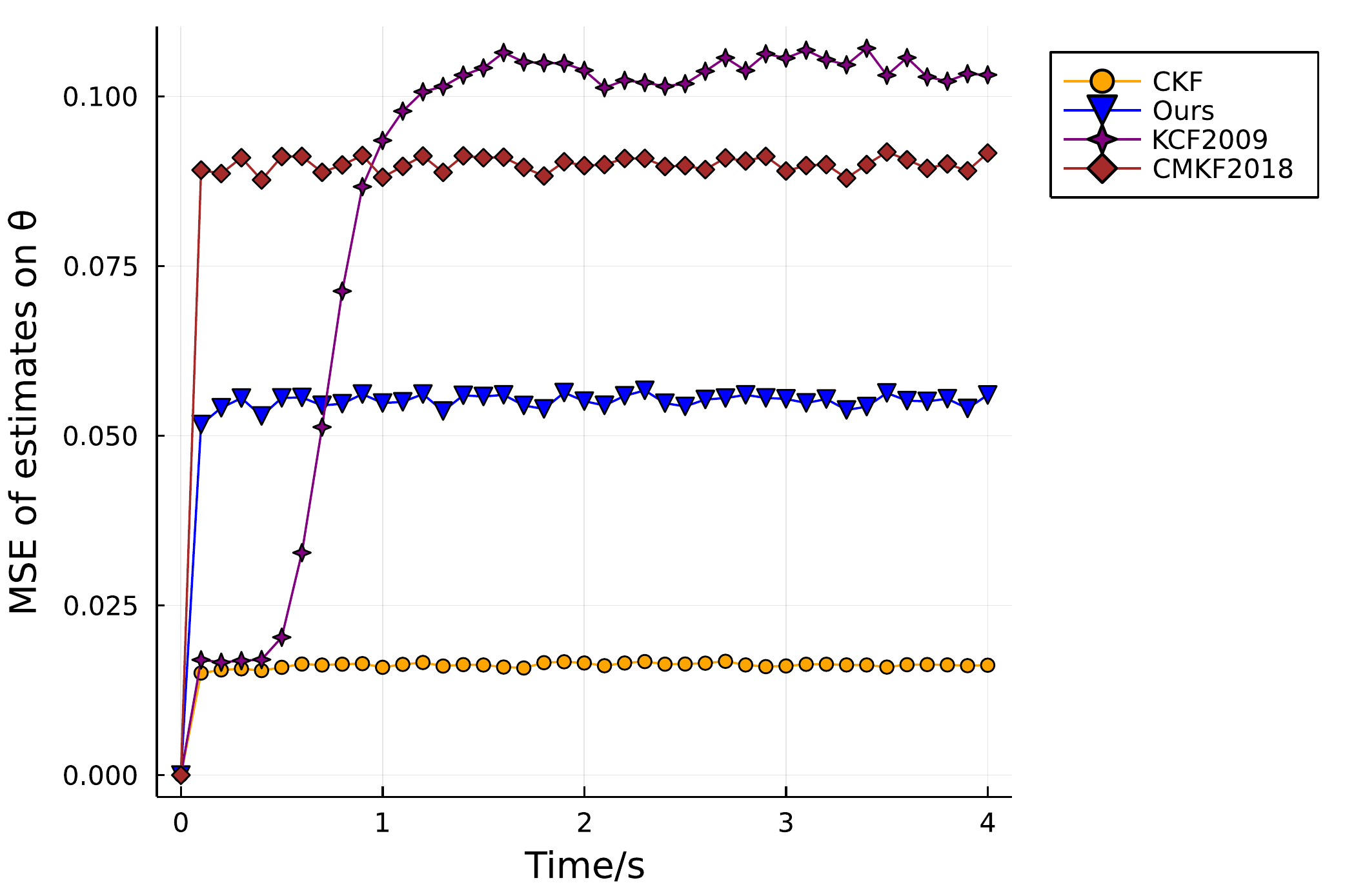}
	\end{minipage}
	\caption{Comparison of the mean square error of the estimates provided by different algorithms in 10000 experiments.}
	\label{fig:compare}
\end{figure}

\subsection{Experiment when the global knowledge on system matrix is unavailable}
Finally, notice that the proposed distributed estimator is based on a lossless decomposition of Kalman filter as developed in Section~\ref{sec:decompose}, which requires the global knowledge on 1) the system matrix $A$, 2) the measurement matrix $C$, and 3) noise covariance matrix $Q$ and $R$. In the case that certain part of $A$, $C$, $Q$ and $R$ are unknown, before running Algorithm~\ref{alg:unobservable} or \ref{alg:n<m}, each sensor can broadcast its local parameters. In this way, every sensor can obtain the system parameters it needs within finite steps. 

To quantify the overhead incurred by this initialization, i.e., broadcasting the parameters, in the third example, we conduct an experiment using $m=15$ raspberry pis equipped with temperature sensors which run the proposed distributed estimation algorithm every minute. In our experiment, it is assumed that the sensors do not have global information on $C$ and $R$. Thus, let each of them broadcast its $C_i$ and $R_i$ at the starting phase so that every sensor can obtain system parameters it needs.

The mean traffic of a sensor with $3$ neighbors is shown in Fig.~\ref{fig:traffic}. It turns out, compared with the centralized Kalman filter, our solution induces lower communication burden even with the additional effort on initial broadcasting. Obviously, the merits become more apparent with the increasing scale of sensor networks.

\begin{figure}[!h]
	\centering
	\includegraphics[width=0.47\textwidth]{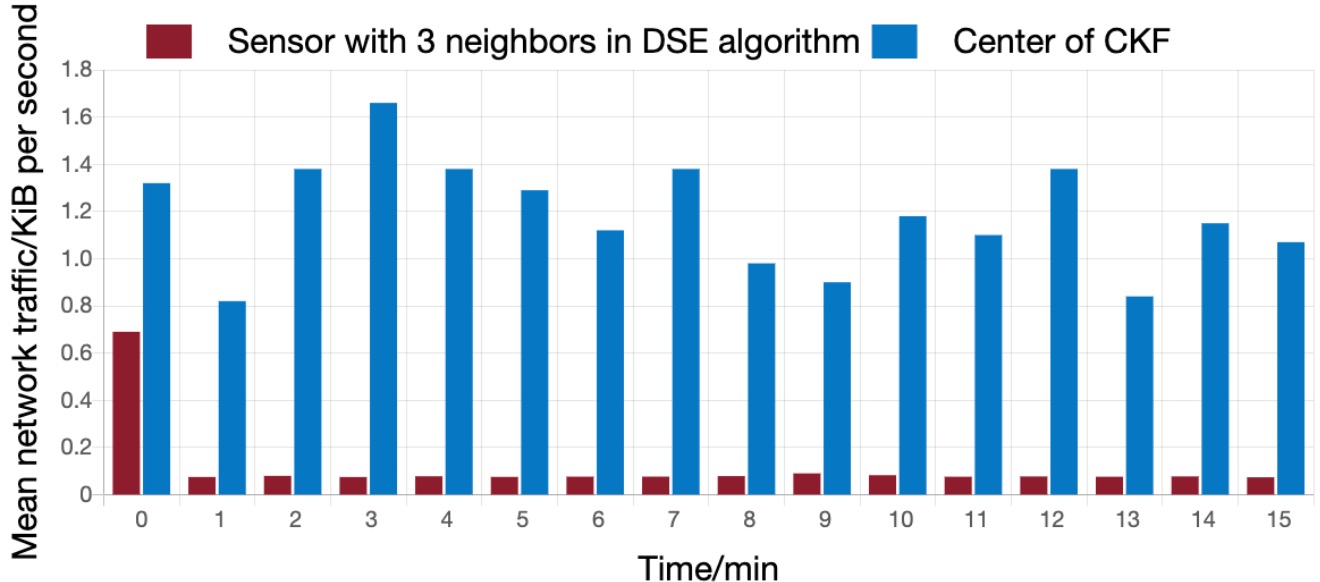}
	\caption{Mean network traffic v.s. time.}
	\label{fig:traffic}
\end{figure}

\section{Conclusion}\label{sec:conclusion}
  
In this paper, the problem of distributed state estimation has been studied for an LTI Gaussian system. We investigate both cases where $m>n$ and $m\leq n$, and propose distributed estimators for both cases to introduce low communication cost. The local estimator is proved to be stable at each sensor side, in the sense that the covariance of estimation error is proved to be bounded and the asymptotic error covariance can also be derived. Our major merit lays in reformulating the problem of distributed estimation to that of linear system synchronization. 

\appendices

\section{Proof of Lemma~\ref{lmm:observable}}
We will prove by contradiction. If $(X^T+qp^T,q)$ is not controllable, then we can find some $s$, such that the rank of $\begin{bmatrix}
		X^T+qp^T-sI & q
	\end{bmatrix}$ is strictly less than $n$. Therefore, there exists a non-zero $v$, such that
	\[v^T\begin{bmatrix}
		X^T+qp^T-sI & q
	\end{bmatrix} = 0,\]
	which implies that
	\[
	(X+pq^T)v - sv = 0,\,q^T v = 0.
	\]
Therefore $(X+pq^T)v - sv = 0$ and $Xv -sv = 0$, implying that $s$ is an eigenvalue of both $X$ and $X+pq^T$, which contradicts with the assumption $X$ and $X+pq^T$ do not share eigenvalues. We thus complete the proof.

\section{Proof of Lemma~\ref{lmm:regular}}
    We will prove this lemma by construction. Towards the end, let us next consider the following equation:
	\begin{equation}\label{eqn:T2}
		T [p, Xp, \cdots, X^{n-1}p] = T R_X= [q, Yq, \cdots, Y^{n-1}q] = R_Y,
	\end{equation}
	where $R_X= [p, Xp, \cdots, X^{n-1}p]$ and $R_Y = [q, Yq, \cdots, Y^{n-1}q] $.
	
	Since $(X,p)$ is controllable, $R_X$ is full rank and thus invertible, and $T = R_YR_X^{-1}$ solves \eqref{eqn:T2}. Clearly $Tp = q$. In what follows, we shall prove that $TX=YT$. To this end, let us denote the characteristic polynomial of $X$ as $\varphi(s) = s^n + \alpha_{n-1}s^{n-1} + \ldots \alpha_0$.
	It is noted that
	\begin{equation}
		\begin{split}
			T X^np &= T(-\alpha_{n-1}X^{n-1}-\alpha_{n-2}X^{n-2}-\cdots-\alpha_0I)p\\&=(-\alpha_{n-1}Y^{n-1}q-\alpha_{n-2}Y^{n-2}q-\cdots-\alpha_0q)=Y^n q,
		\end{split}
	\end{equation}
	where the first and the last equality is due to Carley-Hamilton and the second equality is from the fact $TR_X = R_Y$. As a result
	\[
	TXR_X= T [Xp, \cdots, X^{n}p] =  [ Yq, \cdots, Y^{n}q] = YR_Y,
	\]
	Hence, $TX = YR_YR_X^{-1} = YT$, which finishes the proof.

\section{Proof of Lemma~\ref{lmm:localeqv}}\label{sec:app_localeqv}

	1) From \eqref{eqn:xi_z}, it is easy to verify that
	\begin{equation}
		\begin{split}
			&S \hat{\xi}_i(k)+\1_n z_i(k)\\=&(\Lambda+\1_n \beta^T) \hat{\xi}_i(k)+\1_n [y_i(k+1)-\beta^T\hat \xi_i(k)]\\=&\Lambda\hat\xi_i(k)+ \1_ny_i(k+1).
		\end{split}
	\end{equation} 
	As a result, the local filter \eqref{eqn:xi_z} has the same input-output relationship with \eqref{eqn:xi}.
	
	2) By Lemma~\ref{lmm:regular}, we know that for any $i\in\mathcal{V}$, we can find $G_i^u\in\mathbb{R}^{n\times n^u}$, such that
	\begin{equation*}
		(G_i^u)^T S^T = \left(A^u\right)^T (G_i^u)^T,\, (G_i^u)^T\beta = (C_i^uA^u)^T,
	\end{equation*}
	which implies that
	\begin{equation}
		\begin{aligned}
			G_i^uA^u - \mathbf 1_n C_i^uA^u &= SG_i^u - \1_n \beta^T G_i^u \\ &= (\Lambda + 1\beta^T)G_i^u - \1_n \beta^T G_i^u=\Lambda G_i^u  ,\\
			\beta^T G_i^u &= C_i^uA^u.
		\end{aligned}
	\end{equation}
	Furthermore,  
	\begin{equation}\label{eqn:G_i}
		\begin{aligned}
			\begin{bmatrix}
				G_i^u & 0
			\end{bmatrix}A-\mathbf1_nC_{i}A&=	\begin{bmatrix}
				G_i^uA^u & 0
			\end{bmatrix}-\1_{n}\begin{bmatrix}
				C_i^uA^u & C_i^sA^s
			\end{bmatrix}
			\\&=\Lambda\begin{bmatrix}
				G_i^u & 0
			\end{bmatrix}
			- \mathbf 1_{n}\begin{bmatrix}
				0 &C_i^sA^s
			\end{bmatrix}
			,\\
			\beta^T\begin{bmatrix}
				G_i^u & 0
			\end{bmatrix}&= \begin{bmatrix}
				C_i^uA^u &0	
			\end{bmatrix} = C_{i}A -
			\begin{bmatrix}
				0 &
				C_{i}^sA^s
			\end{bmatrix},
		\end{aligned}
	\end{equation}
	where $A$ and $C_i$ are given in \eqref{eqn:diagA} and \eqref{eqn:diagC}, respectively.
	
	For simplicity, we denote 
	\begin{equation}\label{eqn:defGi}
		G_i \triangleq \begin{bmatrix}
			G_i^u & 0
		\end{bmatrix}\in\mathbb{R}^{n\times n}.
	\end{equation}
	Moreover, let 
	\begin{equation}\label{eqn:def_e}
		\epsilon_i(k)\triangleq G_ix(k)-\hat\xi_i(k).
	\end{equation} 
	It follows from \eqref{eqn:xi} that
	\begin{equation}\label{eqn:e_i}
		\begin{split}
			&\quad\;\; \epsilon_i(k+1) = G_ix(k+1)-\hat\xi_i(k+1)\\
			&=G_iAx(k)+G_iw(k)-\Lambda\hat\xi_i(k)- \1_ny_i(k+1)\\
			&=(G_i-\1_nC_i)Ax(k)-\Lambda\hat\xi_i(k)+(G_i-\1_nC_i)w(k)\\&\quad-\1_nv_i(k+1)\\
			&=\Lambda G
			-\1_{n}\begin{bmatrix}
				0 &
				C_{i}^sA^s
			\end{bmatrix}x(k)-\Lambda\hat\xi_i(k)+(G_i-\1_nC_i)w(k)\\&\quad-\1_nv_i(k+1)\\
			&=\Lambda \epsilon_i(k)-\1_{n}C_{i}^sA^sx^s(k) +(G_i-\1_nC_i)w(k)-\1_nv_i(k+1),
		\end{split}
	\end{equation}
	where the second to last equality holds by \eqref{eqn:G_i}.
	Due to the fact that $\Lambda$ and $A^s$ are stable, we conclude that $\epsilon_i(k)$ is stable, i.e., $\cov(\epsilon_i(k))$ is bounded.
	
	One thus has
	\begin{equation}\label{eqn:z_i}
		\begin{split}
			z_i(k)&=y_i(k+1)-\beta^T\hat \xi_i(k)\\&=y_i(k+1)-\beta^T(G_ix(k)-\epsilon_i(k))\\
			&=C_i(Ax(k)+w(k))+v_i(k+1)+\beta^T\epsilon_i(k)\\&\qquad-(C_{i}A -
			\begin{bmatrix}
				0 &
				C_{i}^sA^s
			\end{bmatrix})x(k)\\
			&=\beta^T\epsilon_i(k)+C_{i}^sA^sx^s(k)+C_iw(k)+v_i(k+1).
		\end{split}
	\end{equation}
	\revise{As proved in \eqref{eqn:e_i}, $\cov(\epsilon_i(k))$ is bounded. Moreover, it follows from 
	\eqref{eqn:xs} that $C_{i}^sA^sx^s(k)$ is a linear combination of the stable parts in $x(k)$. Also, the covariance of $w(k)$ and $v_i(k+1)$ are bounded as $Q$ and $R_i$, respectively.
		We thus conclude that $z_i(k)$ is stable, i.e., the covariance of $z_i(k)$ is always bounded. }

\section{Proof of Lemma~\ref{lmm:matrixdecompse}}\label{sec:app_matrixdecomp}
For the proof of Lemma~\ref{lmm:matrixdecompse}, we need the following result:
\begin{lemma}\label{lmm:vec_poly}
	Given any vector $w \in\mathbb{R}^n$. Suppose $(S^T, v)$ is controllable, then there exists a polynomial $p$ of at most $n-1$ degree, such that $w$ can be decomposed as
	\begin{equation}\label{eqn:vectorw}
		w^T=v^T \varphi(S).
	\end{equation}
\end{lemma}
\begin{proof}
	Suppose $\varphi(S)=\alpha_0 I + \alpha_1 S + \cdots +\alpha_{n-1} S^{n-1}$. We thus rewrite \eqref{eqn:vectorw} as
	\begin{equation}
		w=  \begin{bmatrix}
			v& S^Tv& \cdots &\left( S^{n-1}\right)^Tv
		\end{bmatrix}\begin{bmatrix}
			\alpha_0\\
			\vdots\\
			\alpha_{n-1}
		\end{bmatrix}.
	\end{equation}
	Since $(S^T, v)$ is controllable, the first matrix on the RHS of the equation has a column rank of $n$ and hence the above equation is always solvable. We therefore complete the proof.
\end{proof}

Now we are ready to prove Lemma~\ref{lmm:matrixdecompse}. Notice that any matrix $W$ can be decomposed as
\begin{equation}\label{eqn:matrixdecompose}
	W = \begin{bmatrix}
		w_1^T\\\vdots \\w_n^T
	\end{bmatrix}= e_1 w_1^T + e_2 w_2^T + \cdots e_n w_n^T.
\end{equation}
Since $(S^T, \beta)$ is controllable, \eqref{eqn:Wdecompse} can be concluded by applying Lemma \ref{lmm:vec_poly} to \eqref{eqn:matrixdecompose}.

\section{Proof of Theorem~\ref{lmm:mdl_reduce}}\label{sec:app_mdl_reduce}
To begin with, we note that the following relation holds true at any $k\geq 0$:
\begin{equation}
	F_i S^k =  \Big[\sum_{j=1}^n H_j p_{ij}(S) \Big]S^k= \sum_{j=1}^n H_j S^k p_{ij}(S), 
\end{equation}
where the last equality holds as $S$ is commutable with any polynomials of itself. Then let us consider the output of system \eqref{eqn:fullmodel}:
\begin{equation}
	\begin{aligned}
		\hat x(k+1)&=\sum_{t=0}^k F(I_m\otimes S)^t(I_m\otimes \textbf{1}_n)z(k-t)\\
		&= \sum_{t=0}^{k} \Big(\sum_{i=1}^{m}F_iS^t\textbf{1}_nz_i(k-t)\Big) \\
		&= \sum_{t=0}^{k} \Big(\sum_{i=1}^{m}\sum_{j=1}^n H_j S^t p_{ij}(S)\textbf{1}_nz_i(k-t)\Big) \\
		&= \sum_{t=0}^{k} \Big(\sum_{j=1}^n H_j S^t \big[\sum_{i=1}^{m} p_{ij}(S)\textbf{1}_n z_i(k-t)\big]\Big)\\
		&=\sum_{t=0}^{k} H(I_n\otimes S)^t T z(k-t)=\tilde{x}(k+1).
	\end{aligned}
\end{equation}

\revise{Notice that \eqref{eqn:aft_mdl_reduce} has $z(k)$ as its input and $\tilde{x}(k)$ as its output. As proved, given any $z(k)$, \eqref{eqn:fullmodel} and \eqref{eqn:aft_mdl_reduce} yield the same output, i.e., $\tilde{x}(k)=\hat x(k+1)$. Hence, we conclude that the two systems have the identical transfer functions. The proof is thus completed.}

\section{Proof of Theorem \ref{thm:observable}}\label{sec:appB}
For simplicity, we first define aggregated vectors and matrices as below:
\begin{equation}
\begin{split}
\eta(k)&\triangleq 
\begin{bmatrix}
\eta_1(k)\\
\vdots\\
\eta_m(k)
\end{bmatrix},\hat\xi(k)\triangleq \begin{bmatrix}
\hat\xi_1(k)\\
\vdots\\
\hat\xi_m(k)
\end{bmatrix},\\
L_\eta&\triangleq 
\begin{bmatrix}
\tilde L_1 & & \\
& \ddots & \\
& &\tilde L_m
\end{bmatrix}.\\
\end{split}
\end{equation}
Then, we can rewrite \eqref{eqn:eta} in matrix form as:
\begin{equation}\label{eqn:eta_matrix}
\begin{split}
&\eta(k+1)\\=&(I_m\otimes\tilde S)\eta(k)+L_\eta z(k)-[I_m\otimes (\tilde B\tilde \Gamma)](\mathcal{L}_{\mathcal{G}}\otimes I_{n})\eta(k)\\
=&[I_m\otimes \tilde S-\mathcal{L}_{\mathcal{G}}\otimes (\tilde B\tilde\Gamma)]\eta(k)+L_\eta z(k).
\end{split}
\end{equation}

Next let us denote the average state of all agents as 
	\begin{equation}
	\bar{\eta}(k)\triangleq \frac{1}{m} \sum_{i=1}^m \eta_{i}(k)=\frac{1}{m} (\1_m^T\otimes I_{mn})\eta(k).
	\end{equation}
	Since $\1_m^T \mathcal{L}_{\mathcal{G}} =0$, it holds that
	\begin{equation}\label{eqn:average_matrix}
	\begin{split}
	&\bar{\eta}(k+1) \\= &\frac{1}{m} (\1_m^T\otimes I_{mn})\Big([I_m\otimes \tilde S-\mathcal{L}_{\mathcal{G}}\otimes (\tilde B\tilde\Gamma)]\eta(k)+L_\eta z(k)\Big)\\=&\tilde{S}\bar{\eta}(k)+\frac{1}{m}(\1_m^T\otimes I_{mn})L_\eta z(k).
	\end{split}
	\end{equation}
	Furthermore, we define the state deviation of each sensor as $\delta_i(k)\triangleq \eta_i(k)-\bar\eta(k)$ and then stack them as an aggregated vector $\delta(k)\triangleq \col(\delta_1(k), \cdots, \delta_m(k))$.  Combining \eqref{eqn:eta_matrix} and \eqref{eqn:average_matrix} yields the dynamic equation of $\delta(k)$:
	\begin{equation}\label{eqn:error_matrix}
	\begin{split}
	\delta(k+1)&=[I_m\otimes \tilde S-\mathcal{L}_{\mathcal{G}}\otimes (\tilde B\tilde\Gamma)]\delta(k)+L_\delta z(k),
	\end{split}
	\end{equation}
	where \begin{equation}\label{eqn:L_delta}
		L_\delta \triangleq [(I_{m}-\frac{1}{m}\1_m\1_m^T)\otimes I_{mn}]L_\eta.
	\end{equation}
	
	Recall that the Laplacian matrix of an undirected graph is symmetric. Therefore, we can always find an unitary matrix $\Phi\triangleq[\frac{1}{\sqrt{m}}\1_m,\phi_2,\cdots,\phi_m]$, such that $\mathcal{L}_{\mathcal{G}}$ is diagonalized as
	\begin{equation}
	\diag(0,\mu_2,\cdots,\mu_m)=\Phi^T\mathcal{L}_{\mathcal{G}}\Phi.
	\end{equation}
	Using the property of Kronecker product yields that
	\begin{equation}
	\begin{split}
	(\Phi \otimes I_{mn})^T[I_m\otimes \tilde S-\mathcal{L}_{\mathcal{G}}\otimes (\tilde B\tilde\Gamma)](\Phi \otimes I_{mn})\\
	=\diag(\tilde S,\tilde S-\mu_2\tilde B\tilde\Gamma,...,\tilde S-\mu_m\tilde B\tilde\Gamma).
	\end{split}
	\end{equation}
	
	Denote 
	\begin{equation}\label{eqn:delta}
	\tilde\delta(k)\triangleq(\Phi\otimes I_{mn})^T\delta(k).
	\end{equation}
	One has
	\begin{equation}\label{eqn:tildedelta}
	\begin{split}
	\tilde\delta(k+1)=A_{\tilde\delta} \tilde\delta(k)+L_{\tilde\delta} z(k),
	\end{split}
	\end{equation}
	where $A_{\tilde\delta} \triangleq diag(\tilde S,\tilde S-\mu_2\tilde B\tilde\Gamma,\cdots,\tilde S-\mu_m\tilde B\tilde\Gamma)$ and $L_{\tilde\delta} \triangleq [(\Phi^T-\frac{1}{m}\Phi^T\1_m\1_m^T)\otimes I_{mn}]L_\eta$. 
	
	We next study the stability of above system. To proceed, let us partition the state into two parts, i.e., $\tilde\delta(k)=[\tilde\delta^T_1(k),\tilde\delta^T_2(k)]^T$, where $\tilde\delta_1(k)\in\mathbb{R}^{mn}$ is a vector consisting of the first $mn$ entries of $\tilde\delta(k)$ and satisfies 
	\begin{equation}
	\tilde\delta_1(k) = \frac{1}{\sqrt{m}}\sum_{i=1}^m \delta_i(k) = \frac{1}{\sqrt{m}}\sum_{i=1}^m (\eta_i(k)-\bar\eta(k))=0.
	\end{equation}  
	Therefore, $\tilde\delta_1(k)$ is stable. Moreover, it holds that
	\begin{equation}\label{eqn:delta2}
	\tilde\delta_2(k+1) = \diag(\tilde S-\mu_2\tilde B\tilde\Gamma,\cdots,\tilde S-\mu_m\tilde B\tilde\Gamma)\tilde\delta_2(k)+\tilde{L}_{\tilde\delta} z(k),
	\end{equation}
	where $\tilde{L}_{\tilde\delta}$ consists the last $(m^2n-mn)$ rows of $\tilde{L}_{\tilde\delta}$. 
	\revise{In view of Lemma \ref{lmm:poleplacement}, $\diag(\tilde S-\mu_2\tilde B\tilde\Gamma,\cdots,\tilde S-\mu_m\tilde B\tilde\Gamma)$ is Schur. Recalling Lemma~\ref{lmm:localeqv}, $z(k)$ is also stable. We therefore conclude that \eqref{eqn:delta2} is stable, which further implies the stability of \eqref{eqn:tildedelta}. }
	
	On the other hand, one derives from \eqref{eqn:z_i} that
	\revise{
	\begin{equation}
	z(k) = Cw(k)+v(k+1)+(I_m\otimes \beta^T)\epsilon(k)+C^sA^sx^s(k),
	\end{equation}}
	where $\epsilon(k)\triangleq \col(\epsilon_1(k), \cdots, \epsilon_m(k))$ and
	\begin{equation*}
		C^s = {\left[\begin{array}{c}
				C_{1}^s \\
				\vdots \\
				C_{m}^s
			\end{array}\right]}.
	\end{equation*} Recalling \eqref{eqn:e_i}, it follows that
\revise{
\begin{equation*}
\begin{split}
\epsilon(k+1) = (I_m\otimes \Lambda)\epsilon(k)&+ 
\begin{bmatrix}
G_1-\1_n C_1\\
\vdots\\
G_m-\1_n C_m
\end{bmatrix}w(k)\\-(I_m\otimes \1_n) &v(k+1)-(I_m\otimes \1_n) C^sA^sx^s(k)\\=(I_m\otimes \Lambda) \epsilon(k)&+W_\epsilon w(k) + V_\epsilon v(k+1)+A_\epsilon x^s(k),
\end{split}
\end{equation*}
where \begin{equation}
\begin{split}
W_\epsilon &\triangleq \begin{bmatrix}
	G_1-\1_n C_1\\
	\vdots\\
	G_m-\1_n C_m
\end{bmatrix}, \\V_\epsilon&\triangleq-(I_m\otimes \1_n), A_\epsilon\triangleq -(I_m\otimes \1_n) C^sA^s.
\end{split}
\end{equation}}
	
	By combining the above dynamics with \eqref{eqn:xs}, one derives that
	\begin{equation}
	\begin{split}
	\begin{bmatrix}
	\tilde\delta(k+1)\\
	\epsilon(k+1)\\
	x^s(k+1)
	\end{bmatrix}=  &\begin{bmatrix}
	A_{\tilde\delta} & L_{\tilde{\delta}}(I_m\otimes \beta^T) &L_{\tilde{\delta}}C^sA^s\\
	 & I_m\otimes \Lambda &A_\epsilon\\
	 &  & A^s
	\end{bmatrix}\begin{bmatrix}
	\tilde\delta(k)\\
	\epsilon(k)\\
	x^s(k)
	\end{bmatrix}\\&+\begin{bmatrix}
	L_{\tilde{\delta}} C\\
	W_\epsilon\\
	J
	\end{bmatrix}w(k)+\begin{bmatrix}
	L_{\tilde{\delta}} \\
	V_\epsilon\\
	0
	\end{bmatrix}v(k+1)
	\end{split}
	\end{equation}
	
	\revise{Notice that the above system is stable}. Hence, we calculate the covariance at both sides and in steady state. It holds that $W_r$, the steady state covariance, is the unique solution of below Lyapunov equation:
	\begin{equation}\label{eqn:Lyp1}
	\begin{split}
	W_r = A_rW_rA_r^T + \begin{bmatrix}
	L_{\tilde{\delta}} C\\
	W_\epsilon\\
	J
	\end{bmatrix} Q \begin{bmatrix}
	L_{\tilde{\delta}} C\\
	W_\epsilon\\
	J
	\end{bmatrix}^T+ \begin{bmatrix}
	L_{\tilde{\delta}} \\
	V_\epsilon\\
	0
	\end{bmatrix} R \begin{bmatrix}
	L_{\tilde{\delta}} \\
	V_\epsilon\\
	0
	\end{bmatrix}^T,
	\end{split}
	\end{equation}
where $$A_r=\begin{bmatrix}
	A_{\tilde\delta} & L_{\tilde{\delta}}(I_m\otimes \beta^T) &L_{\tilde{\delta}}C^sA^s\\
	& I_m\otimes \Lambda &A_\epsilon\\
	&  & A^s
\end{bmatrix}.$$
	
	In view of \eqref{eqn:delta}, it holds that
	\begin{equation}
	\delta(k) = \begin{bmatrix}
	\Phi \otimes I_{mn} & 0 &0
	\end{bmatrix}
	\begin{bmatrix}
	\tilde\delta(k)\\
	\epsilon(k)\\
	x^s(k)
	\end{bmatrix}=\Phi_\delta \begin{bmatrix}
	\tilde\delta(k)\\
	\epsilon(k)\\
	x^s(k)
	\end{bmatrix},
	\end{equation}
where \begin{equation}
\Phi_\delta \triangleq \begin{bmatrix}
	\Phi \otimes I_{mn} & 0 &0
\end{bmatrix}.
\end{equation}
	Moreover, let us denote
	\begin{equation}
	\bar{e}_i(k) \triangleq \breve{x}_i(k)-\hat{x}(k),
	\end{equation}
	which is the bias from local estimate $\breve{x}_i(k)$ to optimal Kalman one. Combining \eqref{eqn:localdecompose} and \eqref{eqn:xvseta} yields
	\begin{equation}
	\hat{x}(k) = F\sum_{i=1}^m \eta_i(k).
	\end{equation}
	One thus has
	\begin{equation}\label{eqn:bar_e}
	\bar{e}_i(k) = mF(\eta_i(k)-\bar{\eta}(k))=mF\delta_i(k).
	\end{equation}
	Stacking such errors from all sensors together yields
	\begin{equation}
	\begin{split}
	\bar{e}(k) = (I_m\otimes mF) \delta(k)=(I_m\otimes mF)\Phi_\delta \begin{bmatrix}
	\tilde\delta(k)\\
	\epsilon(k)
	\end{bmatrix}.
	\end{split}
	\end{equation}
	Therefore, in steady state,  the covariance of $\bar{e}(k)$ can be calculated as
	\begin{equation}\label{eqn:Lyp2}
	\bar W = [(I_m\otimes mF)\Phi_\delta]W_r[(I_m\otimes mF)\Phi_\delta]^T.
	\end{equation}
	Finally, for any sensor $i$, let us denote its estimation error as
	\begin{equation}\label{eqn:breve_e}
	\begin{split}
	\breve{e}_i(k) &= \breve{x}_i(k)-x(k)\\&=(\breve{x}_i(k)-\hat{x}(k))+(\hat{x}(k)-x(k)) \\&= \bar{e}_i(k)+\hat{e}(k),
	\end{split}
	\end{equation}
	where $\hat{e}(k)$ is the estimation error of Kalman filter. Since Kalman filter is optimal, $\bar{e}_i(k)$ is orthogonal to $\hat{e}(k)$.
	
	By defining $\breve{e}(k)\triangleq \col(\breve{e}_1(k) , \cdots, \breve{e}_m(k))$, we therefore have
	\begin{equation}
	\breve{e}(k) = \bar{e}(k)+\1_m \otimes \hat{e}(k).
	\end{equation}
	Calculating the covariance of both sides yields
	\begin{equation}\label{eqn:covariance}
	\breve{W} = \bar W + (\1_m\1_m^T )\otimes P,
	\end{equation}
	where $	\breve{W}$ is the steady-state covariance of $\breve{e}(k)$ and $P$ is given in \eqref{eqn:KFcov}. Notice that the above calculation also indicates the boundedness of $\cov(\breve{e}(k))$ at any time.

\modify{\section{Proof of Corollary~\ref{col:error}}
As proved in Appendix-\ref{sec:appB}, one can exactly calculate $\bar{W}$ by solving Lyapunov equations \eqref{eqn:Lyp1} and \eqref{eqn:Lyp2}. The result is thus obvious by invoking \eqref{eqn:covariance}.}

\section{Proof of Theorem~\ref{thm:general}}\label{sec:appC}
To proceed, let us introduce the following lemma:
\begin{lemma}\label{lmm:cauchy}
	Given any random variables $\kappa_1,\cdots,\kappa_\tau$, it follows that
	\begin{equation}\label{eqn:cauchy}
	\mathbb{E}\Big[\Big|\Big|\sum_{i=1}^\tau \kappa_i\Big|\Big|^2\Big]\leq
	\Big(\sum_{i=1}^\tau \sqrt{\mathbb{E}[||\kappa_i||^2]}\Big)^2.
	\end{equation}
\end{lemma}
\begin{proof}
	In order to prove \eqref{eqn:cauchy}, it is equivalent to show that
	\begin{equation}
	\sum_{i=1}^\tau\sum_{j=1}^\tau \mathbb{E}[\kappa_i^T\kappa_j] \leq \sum_{i=1}^\tau\sum_{j=1}^\tau \sqrt{\mathbb{E}[\kappa_i^T\kappa_i]} \sqrt{\mathbb{E}[\kappa_j^T\kappa_j]}.
	\end{equation}
	By Cauchy-Schwarz inequality, it holds for any $i,j$ that
	\begin{equation}
	\mathbb{E}[\kappa_i^T\kappa_j] \leq \sqrt{\mathbb{E}[\kappa_i^T\kappa_i]} \sqrt{\mathbb{E}[\kappa_j^T\kappa_j]}.
	\end{equation}
	The proof is thus completed.
\end{proof}

We next prove Theorem~\ref{thm:general}. Applying similar arguments to Theorem~\ref{thm:optimal}, it is easy to see from the consistency condition \eqref{eqn:consistency} that the average of local estimates coincides with the optimal Kalman filter. We hence focus on the analysis of estimation error covariance.

Let us denote $\delta_i(k)\triangleq \eta_i(k)-1/m \sum_{i=1}^m \eta_i(k)$ and $\varpi_i(k)\triangleq\omega_i(k)-1/m \sum_{i=1}^m \omega_i(k)$. 
Moreover, we define \begin{align*}
		&\delta(k) \triangleq \col(\delta_1(k),\cdots,\delta_m(k)),\\
		&\varpi(k)\triangleq \col(\varpi_1(k),\cdots,\varpi_m(k)).
\end{align*}
It hence follows from \eqref{eqn:general} that
\begin{equation}
\begin{bmatrix}
\delta(k+1)\\
\varpi(k)
\end{bmatrix}=\begin{bmatrix}
\mathcal{D}(k) & \mathcal{J}(k)\\
\widetilde{\mathcal{B}} & \widetilde{\mathcal{A}} 
\end{bmatrix}\begin{bmatrix}
\delta(k)\\
\varpi(k-1)
\end{bmatrix}+\begin{bmatrix}
L_\delta\\
0
\end{bmatrix}z(k),
\end{equation}
where $L_\delta$ is defined in \eqref{eqn:L_delta}, and
\begin{align*}
&\mathcal{D}(k) \triangleq I_m\otimes \tilde{S}-\mathcal{L}(k)\otimes (\tilde{B} \tilde{\Gamma}\mathcal{B}),\\
&\mathcal{J}(k) \triangleq -\mathcal{L}(k)\otimes (\tilde{B} \tilde{\Gamma}\mathcal{A}),
\widetilde{\mathcal{A}} \triangleq I_{m} \otimes \mathcal{A},\;\widetilde{\mathcal{B}} \triangleq I_{m} \otimes \mathcal{B},
\end{align*}
with $\mathcal{L}(k)\triangleq\{\mathcal{L}_{i,j}(k)\}$ being the (random) Laplacian matrix with respect to the weights $\{a_{ij}\gamma_{ij}(k)\}$. Namely,
\begin{equation}
	\mathcal{L}_{i,j}(k) \triangleq
	\left\{\begin{array}{l}{\sum_{l=1}^m a_{il}\gamma_{il}(k),\;\quad j= i} \\ {-a_{ij}\gamma_{ij}(k), \qquad\quad\; j\ne i }  \end{array}.\right.
\end{equation}
For simplicity, Let 
\begin{align}
\mathcal{Q}(k) \triangleq \begin{bmatrix}
\mathcal{D}(k) & \mathcal{J}(k)\\
\widetilde{\mathcal{B}} & \widetilde{\mathcal{A}} 
\end{bmatrix}.
\end{align}
Since $\delta_i(0)=0$ and $\varpi_i(0)=0$ hold for any $i$, it follows that
\begin{equation}
\begin{bmatrix}
	\delta(k+1)\\
	\varpi(k)
\end{bmatrix}=\sum_{t=0}^{k}\bigg(\mathcal{Q}(k, t+1)\begin{bmatrix}
L_\delta\\
0
\end{bmatrix}z(t)\bigg),
\end{equation}
where the transition matrix is defined as
$$
\mathcal{Q}(k, s)=\left\{\begin{array}{cc}
\mathcal{Q}(k) \mathcal{Q}(k-1) \cdots \mathcal{Q}(s), & k\geq s, \\
I, & k<s .
\end{array}\right.
$$

Then consider the update of any agent $i$. From the above equation, we conclude that 
\revise{\begin{equation}\label{eqn:delta_general}
\delta_i(k+1) = \sum_{t=0}^{k}\Pi_i(k, t+1)z(t),
\end{equation}}
where $\Pi_i(k, t+1)$ refers to the $i$-th row of matrix $\mathcal{Q}(k, t+1)[L_\delta \quad 0]^T.$
Namely, the consensus error of agent $i$, i.e. $\delta_i(k+1)$, is caused by the sequence of residuals $\{z(t)\}$, where $t\leq k$. For simplicity, we denote
$$\kappa_i(k,t) \triangleq\Pi_i(k, t+1)z(t).$$ \revise{Since $\cov(z(t))$ is bounded at any time, in view of \eqref{eqn:consensus}, the following statement holds for any $t\leq k$:
\begin{equation}
\mathbb{E}[||\kappa_i(k,t)||^2] \leq c\rho^{k-t}.
\end{equation}}
Therefore, one has that
\begin{equation}
\begin{split}
cov(\delta_i(k+1)) &= \mathbb{E}[||\delta_i(k+1)||^2]=\mathbb{E}\Big[\Big|\Big|\sum_{t=0}^{k}\kappa_i(k,t)\Big|\Big|^2\Big]\\& \leq \Big(\sum_{i=1}^\tau \sqrt{\mathbb{E}[||\kappa_i(k,t)||^2]}\Big)^2\leq\bigg(\sum_{t=0}^{k}\sqrt{c\rho^{k-t}}\bigg)^2\\&=\frac{c(1-\sqrt{\rho}^k)^2}{(1-\sqrt{\rho})^2},
\end{split}
\end{equation}
where the first inequality holds by using Lemma~\ref{lmm:cauchy}. Since $\rho\in(0,1)$, combining the above results with \eqref{eqn:bar_e} and \eqref{eqn:breve_e} yields that the estimation error is stable.

\revise{
	\begin{remark}\label{rmk:z}
	It is noted that the reformulation \eqref{eqn:xi_z} with stable input $z_i(k)$ is essential to establish the stability of local estimators. To be concrete, the stability of \eqref{eqn:delta_general} is guaranteed under the bounded input, which is key to prove the boundedness of estimation error covariance, as can be observed from \eqref{eqn:bar_e}-\eqref{eqn:covariance}.
	On the other hand, if an unstable input, e.g., $y_i(k)$ as in \eqref{eqn:xi}, is applied, we cannot conclude on the stability of local estimator even using the exponentially converged synchronization algorithms. 
\end{remark}}

  \bibliographystyle{IEEEtran}
  \bibliography{reference}
  \end{document}